\documentclass[english,numberwithinsect]{lipics-v2019}

\usepackage{cite} \usepackage{microtype} 

\usepackage{comment}
\usepackage{xspace}
\usepackage{mathtools}

\newcommand{\ob}[1]{\mkern 1.5mu\overline{\mkern-1.5mu#1\mkern-1.5mu}\mkern 1.5mu}

\usepackage{caption}

\usepackage{tkz-fct}
\usetikzlibrary{intersections, plotmarks, arrows.meta, angles, quotes}

\tikzset{
        dot/.style={circle,inner sep=1.7pt,fill,label={#1},name=#1},
    right angle quadrant/.code={
        \pgfmathsetmacro\quadranta{{1,1,-1,-1}[#1-1]}     \pgfmathsetmacro\quadrantb{{1,-1,-1,1}[#1-1]}},
    right angle quadrant=1, right angle length/.code={\def\rightanglelength{#1}},   right angle length=2ex, right angle symbol/.style n args={3}{
        insert path={
            let \p0 = ($(#1)!(#3)!(#2)$) in     let \p1 = ($(\p0)!\quadranta*\rightanglelength!(#3)$), \p2 = ($(\p0)!\quadrantb*\rightanglelength!(#2)$) in let \p3 = ($(\p1)+(\p2)-(\p0)$) in  (\p1) -- (\p3) -- (\p2)
        }
    }
}

\expandafter\ifx\csname pdfoptionalwaysusepdfpagebox\endcsname\relax\else
\fi

\newcommand{\BeginMyItemize}{\begin{itemize}\setlength{\itemsep}{-\parskip}}
\newcommand{\EndMyItemize}{\end{itemize}}
\newcommand{\myitemize}[1]{\BeginMyItemize #1 \EndMyItemize}

\newcommand{\BeginMyEnumerate}{\begin{enumerate}\setlength{\itemsep}{-\parskip}}
\newcommand{\EndMyEnumerate}{\end{enumerate}}

\renewcommand{\leq}{\leqslant}
\renewcommand{\geq}{\geqslant}
\newcommand{\mypara}[1]{\vspace{10pt} \noindent \textbf{\sffamily #1}}

\theoremstyle{plain}

\newenvironment{myquote}{\list{}{\leftmargin=4mm\rightmargin=4mm}\item[]}{\endlist}
\newenvironment{claiminproof}{\begin{myquote}\noindent\emph{Claim.}}{\end{myquote}}
\newenvironment{proofinproof}{\begin{myquote}\noindent\emph{Proof.}}{
  \end{myquote}}

\theoremstyle{plain}
\newtheorem{observation}[theorem]{Observation}

\newcommand{\Reals}{{\mathbb R}}

\newcommand{\bd}{\partial}
    \newcommand{\mycl}{\mathrm{cl}}

\newcommand{\conv}{\mathrm{conv}} \newcommand{\VD}{\mathrm{Vor}}  

\newcommand{\tree}{\ensuremath{\mathcal{T}}}

\newcommand{\A}{\ensuremath{\mathcal{A}}}

\newcommand{\C}{\ensuremath{\mathcal{C}}}
\newcommand{\D}{\ensuremath{\mathcal{D}}}

\newcommand{\V}{\ensuremath{\mathcal{V}}}

\newcommand{\myin}{\mathrm{in}}
\newcommand{\myout}{\mathrm{out}}

\newcommand{\myleft}{\mathrm{left}}
\newcommand{\myright}{\mathrm{right}}

\newcommand{\eps}{\varepsilon}
\newcommand{\mydef}{\coloneqq} \newcommand{\etal}{\emph{et al.}\xspace}

\newcommand{\beats}{\succ_{\hspace*{-0.5mm}\beta}}
\newcommand{\isbeatenby}{\prec_{\hspace*{-0.2mm}\beta}}
\newcommand{\med}{\mathrm{med}}

\newcommand{\Alg}{\textsc{Alg}\xspace}

\newcommand{\slope}{\mathrm{slope}}
\newcommand{\region}{\mathrm{region}}
\DeclareMathOperator{\vd}{vd}

\title{On \texorpdfstring{$\beta$}{\textbeta}-Plurality Points in Spatial Voting Games}

\author{Boris Aronov}{
        Tandon School of Engineering, New York University, Brooklyn, NY 11201, USA}{boris.aronov@nyu.edu}{http://orcid.org/0000-0003-3110-4702}{Partially supported by NSF grant CCF-15-40656 and by grant 2014/170 from the US-Israel Binational Science Foundation.}
\author{Mark de Berg}{
        Department of Computing Science, TU Eindhoven, 5600 MB Eindhoven, The Netherlands}{m.t.d.berg@tue.nl}{https://orcid.org/0000-0001-5770-3784}{Supported by the Netherlands' Organisation for Scientific Research (NWO) under project no.~024.002.003.}
\author{Joachim Gudmundsson}{
        School of Computer Science, University of Sydney, Sydney, NSW 2006, Australia}{joachim.gudmundsson@sydney.edu.au}{https://orcid.org/0000-0002-6778-7990}{Supported under the Australian Research Council Discovery Projects funding scheme (project numbers DP150101134 and DP180102870).}
\author{Michael Horton}{
        Sportlogiq, Inc., Montreal, Quebec H2T 3B3, Canada}{michael.horton@sportlogiq.com} {https://orcid.org/0000-0001-6388-9634}{}

\authorrunning{B.~Aronov, M.~de Berg, J.~Gudmundsson, and M.~Horton} 

\Copyright{Boris Aronov, Mark de Berg, Joachim Gudmundsson and Michael Horton}


\EventEditors{Sergio Cabello and Danny Z. Chen}
\EventNoEds{2}
\EventLongTitle{36th International Symposium on Computational Geometry (SoCG 2020)}
\EventShortTitle{SoCG 2020}
\EventAcronym{SoCG}
\EventYear{2020}
\EventDate{June 23--26, 2020}
\EventLocation{Z\"{u}rich, Switzerland}
\EventLogo{socg-logo}
\SeriesVolume{164}
\ArticleNo{XX}

\hideLIPIcs
\nolinenumbers
\begin{document}

\maketitle

\begin{abstract}
Let $V$ be a set of $n$ points in $\Reals^d$, called \emph{voters}.  A point $p\in \Reals^d$
is a \emph{plurality point} for $V$ when the following holds: for every $q\in\Reals^d$
the number of voters closer to $p$ than to $q$ is at least the number of voters
closer to $q$ than to $p$. Thus, in a vote where
each~$v\in V$ votes for the nearest proposal (and voters for which
the proposals are at equal distance abstain), proposal~$p$ will not lose against
any alternative proposal~$q$. For most voter sets a plurality point
does not exist. We therefore introduce the concept of \emph{$\beta$-plurality points}, which are defined similarly
to regular plurality points except that
the distance of each voter to $p$ (but not to~$q$) is scaled by a factor~$\beta$, for
some constant~$0<\beta\leq 1$. We investigate the existence and computation
of $\beta$-plurality points, and obtain the following results.
\begin{itemize}
\item Define $\beta^*_d \mydef \sup \{ \beta : \text{any finite multiset $V$ in $\Reals^d$ admits a $\beta$-plurality point} \}$.
  We prove that $\beta^*_2 = \sqrt{3}/2$, and that
  $1/\sqrt{d} \leq \beta^*_d \leq \sqrt{3}/2$ for all $d\geq 3$.
\item Define $\beta(p, V) \mydef \sup \{ \beta : \text{$p$ is a $\beta$-plurality point for $V$}\}$.  Given a voter set $V \in \Reals^2$, we provide an algorithm that runs in $O(n \log n)$ time and computes a point $p$ such that $\beta(p, V) \geq \beta^*_2$. Moreover, for $d\geq 2$ we can compute a point~$p$ with
        $\beta(p,V) \geq 1/\sqrt{d}$ in $O(n)$ time.
\item Define $\beta(V) \mydef \sup \{ \beta : \text{$V$ admits a $\beta$-plurality point}\}$.
      We present an algorithm that, given a voter set $V$ in $\Reals^d$, computes an $(1-\eps)\cdot \beta(V)$ plurality point in time
      $O(\frac{n^2}{\eps^{3d-2}} \cdot \log \frac{n}{\eps^{d-1}} \cdot \log^2 \frac {1}{\eps})$.
\end{itemize}
\end{abstract}

\section{Introduction}
\label{se:introduction}

\subparagraph{Background.}
Voting theory is concerned with mechanisms to combine  preferences of
individual voters into a collective decision. A desirable property of such a
collective decision is that it is stable, in the sense that no alternative is
preferred by more voters. In spatial voting games~\cite{b-rgdm-48,d-etpad-57} this is formalized as follows; see Fig.~\ref{fig:Plott_and_yolk}(i) for an example in a political context.
The space of all possible decisions is modeled as~$\Reals^d$ and every voter
is represented by a point in $\Reals^d$, where the dimensions represent
different aspects of the decision and the point representing a voter corresponds
to the ideal decision for that voter. A voter $v$ now prefers a proposed decision
$p\in\Reals^d$ over some alternative proposal $q\in\Reals^d$ when~$v$ is closer to~$p$ than to~$q$.
Thus a point $p\in\Reals^d$ represents a stable decision for a given finite set~$V$ of voters
if, for any alternative $q\in\Reals^d$, we have
$\left |\strut \{v\in V : |vp|<|vq|\} \right| \geq \left|\strut\{v\in V : |vq| < |vp| \}\right|$.
Such a point~$p$ is called a \emph{plurality point}.\footnote{One can also
require $p$ to be strictly more popular than any alternative~$q$. This is sometimes
called a \emph{strong} plurality point, in contrast to the \emph{weak} plurality points
that we consider.}

For $d=1$, a plurality point always exists, since in $\Reals^1$ a median of $V$ is
a plurality point. This is not true in higher dimensions, however.
Define a \emph{median hyperplane} for a set $V$ of voters
to be a hyperplane~$h$ such that both open half-spaces defined by $h$ contain fewer
than $|V|/2$ voters. For $d\geq 2$ a plurality point in $\Reals^d$ exists if and only if all median
hyperplanes for~$V$ meet in a common point; see Fig.~\ref{fig:Plott_and_yolk}(ii).
This condition is known as \textit{generalized Plott symmetry} conditions~\cite{em-ppscm-83,plott};
see also the papers by  Wu~\etal~\cite{wlwc-cppcp-13} and  de~Berg~\etal~\cite{bgm-facpp-18},
who present algorithms to determine the existence of a plurality point for a given set of voters.

It is very unlikely that voters are distributed in such a way that all median
hyperplanes have a common intersection. (Indeed, if this happens, then a slightest generic perturbation of a single voter destroys the existence of the plurality point.)
When a plurality point does not exist, we may want to find a point that is
close to being a plurality point. One way to formalize this is to consider
the center of the \emph{yolk} (or \emph{plurality ball}) of~$V$, where the
yolk~\cite{feld-centripetal,gw-cysvg-19,yolk-mckelvey,miller-handbook} is the smallest ball
intersecting every median hyperplane of~$V$.
We introduce $\beta$-plurality points as an alternative way to relax the requirements
for a plurality point, and study several combinatorial and algorithmic questions
regarding $\beta$-plurality points.

  \begin{figure}
  \centering
  \includegraphics[width=12cm]{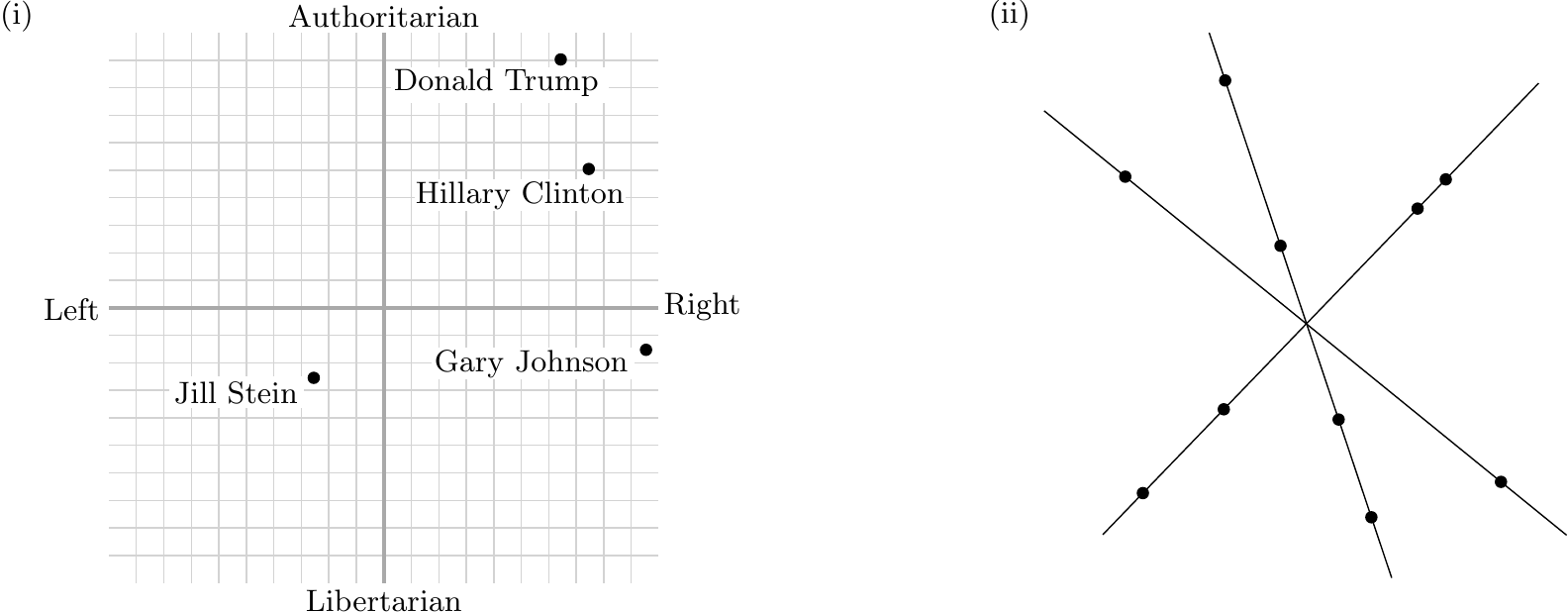}
        \caption{(i) The US presidential candidates 2016 modelled in the spatial voting model, according
                     to The Political Compass (\url{https://politicalcompass.org/uselection2016}).
                     Note that the points representing voters are not shown.
                 (ii) The point set satisfies the generalized Plott symmetry conditions and therefore admits a plurality point.
}
        \label{fig:Plott_and_yolk}
    \end{figure}

\subparagraph{$\beta$-Plurality points: definition and main questions.}
Let $V$ be a multiset\footnote{Even though we allow $V$ to be a multiset,
we sometimes refer to it as a ``set'' to ease the reading. When the fact
that $V$ is a multiset requires special treatment, we explicitly address this.}
of $n$ voters in $\Reals^d$ in arbitrary, possibly coinciding, positions.
In the traditional setting a proposed point~$p\in \Reals^d$ wins a voter $v\in V$
against an alternative~$q$ if $|pv|<|qv|$. We relax this by fixing a parameter~$\beta$
with $0<\beta\leq 1$ and letting $p$ win $v$ against $q$ if $\beta\cdot|pv| < |qv|$.
Thus we give an advantage to the initial proposal~$p$ by scaling distances
to~$p$ by a factor~$\beta\leq 1$.
We now define
\[
V[p\beats q] \mydef \{ v\in V : \beta \cdot |pv| < |qv| \}
\quad \text{and} \quad
V[p\isbeatenby q] \mydef \{ v\in V : \beta \cdot |pv| > |qv| \}
\]
to be the multisets of voters won by $p$ over $q$ and lost by
$p$ against $q$, respectively. Finally, we say that
a point $p\in\Reals^d$ is a \emph{$\beta$-plurality point} for $V$ when
\[
\left |\strut V[p\succ_\beta q] \right | \geq \left|\strut V[p\isbeatenby q]\right|,
    \quad \text{for any point $q\in \Reals^d$}.
\]
Observe that $\beta$-plurality is \emph{monotone} in the sense
that if $p$ is a $\beta$-plurality point then $p$ is also a $\beta'$-plurality point for
all~$\beta' < \beta$.

The spatial voting model was popularised by Black~\cite{b-rgdm-48} and Down~\cite{d-etpad-57} in the 1950s. Stokes~\cite{s-smpc-63} criticized its simplicity and was the first to highlight
the importance of taking non-spatial aspects into consideration.
The reasoning is that voters may evaluate a candidate not only on their policies---their position in
the policy space---but also take their so-called \emph{valence} into account:
charisma, competence, or other desirable qualities in the public's mind~\cite{e-vp-19}.
A candidate can also increase her valence by a stronger party
support~\cite{w-tpsed-05} or campaign spending~\cite{hlm-ppcsv-08}.
Several models have been proposed to bring the spatial model closer to a more realistic voting approach;
see~\cite{supreme,ghr-eavec-11,empirical} as examples. A common model is the multiplicative model,
introduced by Hollard and Rossignol~\cite{hr-aavas-08}, which is closely related to the concept
of a $\beta$-plurality point. The multiplicative model augments the existing spatial utility function
by scaling the candidate's valence by a multiplicative factor. Note that in the 2-player game considered
in this paper the multiplicative model is the same as our $\beta$-plurality model.
From a computational point of view very little is known about the multiplicative model.
We are only aware of a result by Chung~\cite{jonathan}, who studied the problem of positioning
a new candidate in an existing space of voters and candidates, so that the valence required
to win at least a given number of voters is minimized.
\medskip

One reason for introducing $\beta$-plurality was that
a set $V$ of voters in $\Reals^d$, for $d\geq 2$, generally does not admit a plurality point.
This immediately raises the question: Is it true that, for $\beta$ small enough, any set $V$
admits a $\beta$-plurality point? If so, we want to know the largest $\beta$ such that
any voter set~$V$ admits a $\beta$-plurality point, that is, we wish to determine
\[
\beta^*_d \mydef \sup \{ \beta : \text{any finite multiset $V$ in $\Reals^d$ admits a $\beta$-plurality point} \}.
\]
Note that $\beta^*_1 =1$, since any set $V$ in $\Reals^1$ admits a plurality point
and 1-plurality is equivalent to the traditional notion of plurality.

After studying this combinatorial problem in Section~\ref{se:universal}, we turn our attention
to the following algorithmic question: given a voter set $V$, find a point~$p$
that is a $\beta$-plurality point for the largest possible value~$\beta$.
In other words, if we define
\[
\beta(V) \mydef \sup \{ \beta : \mbox{$V$ admits a $\beta$-plurality point} \}
\]
and
\[
\beta(p,V) \mydef \sup \{ \beta : \mbox{$p$ is a $\beta$-plurality point for $V$} \}
\]
then we want to find a point $p$ such that $\beta(p,V)=\beta(V)$.

\subparagraph{Results.}
In Section~\ref{se:universal} we prove that $\beta^*_d\leq \sqrt{3}/2$ for all~$d\geq 2$.
To this end we first show that $\beta^*_d$ is non-increasing in~$d$, and then we
exhibit a voter set~$V$ in~$\Reals^2$ such that~$\beta(V)\leq \sqrt{3}/2$.
We also show how to construct in $O(n \log n)$ time, for any given $V$ in~$\Reals^2$, a point $p$ such that $\beta(p,V) \geq \sqrt{3}/2$,
thus proving that $\beta^*_2=\sqrt{3}/2$.
Moreover, for $d\geq 2$ we show how to construct in $O(n)$ time a point $p$ such that $\beta(p, V) \geq 1/\sqrt{d}$,
which means that $\beta^*_d \geq 1/\sqrt{d}$.\footnote{\label{footnote1}Very recently Filtser and Filtser~\cite{filtser-paper} improved these results for $d\geq 4$ by proving that $\beta^*_d \geq \frac{1}{2}\sqrt{\frac{1}{2}+\sqrt{3}-\frac{1}{2}\sqrt{4\sqrt{3}-3}}\approx 0.557$ for any $d\geq 4$.}

In Section~\ref{se:maximize_beta} we study the problem of computing, for a given
voter set~$V$ of $n$ points in~$\Reals^d$, a $\beta$-plurality point for the
largest possible~$\beta$.
(Here we assume $d$ to be a fixed constant.)
While such a point can be found in polynomial time,
the resulting running time is quite high. We therefore focus our attention
on finding an approximately optimal point~$p$, that is, a point~$p$ such that
$\beta(p,V) \geq (1-\eps)\cdot\beta(V)$. We show that such a point can be
computed in $O(\frac{n^2}{\eps^{3d-2}} \cdot \log \frac{n}{\eps^{d-1}} \cdot \log^2 \frac {1}{\eps})$ time.

\subparagraph{Notation.}
We denote the \emph{open} ball of radius $\rho$ centered at a point $q\in \Reals^d$
by $B(q,\rho)$ and, for a point $p\in\Reals^d$ and a voter $v$, we define
$D_{\beta}(p,v) \mydef B(v,\beta \cdot |pv|)$.
Observe that $p$ wins~$v$ against a competitor~$q$ if and only if $q$ is strictly outside~$D_\beta(p,v)$,
while $q$ wins~$v$ if and only if $q$ is strictly inside~$D_\beta(p,v)$.
Hence, $V[p \prec_\beta q] = \{v \in V \colon q \in D_{\beta}(p,v) \}$.
We define $\D_{\beta}(p) \mydef \{D_{\beta}(p,v) : v \in V\}$---here
we assume $V$ is clear from the context---and
let $\A(\D_{\beta}(p))$ denote the arrangement induced by $\D_{\beta}(p)$.
The competitor point $q$ that wins the most voters against $p$ will thus
lie in the cell of $\A(\D_{\beta}(p))$ of the greatest depth or, more precisely,
the cell contained in the maximum number of disks~$D_{\beta}(p,v)$.

\section{Bounds on $\beta^*_d$}
\label{se:universal}
In this section we will prove bounds on~$\beta^*_d$, the supremum
of all~$\beta$ such that any finite set $V\subset \Reals^d$ admits a $\beta$-plurality point. We start with an observation that allows us to apply bounds on~$\beta^*_d$
to those on~$\beta^*_{d'}$ for $d'>d$.  Let $\conv(V)$ denote the convex hull of~$V$.
\begin{observation}\label{obs:dimension-monotonicity}
Let $V$ be a finite multiset of voters in $\Reals^d$.
\begin{enumerate}
\item[(i)] Suppose a point $p\in \Reals^d$ is not a $\beta$-plurality point for $V$. Then there is a point $q\in\conv(V)$
such that $\left |\strut V[p\succ_\beta q] \right | < \left|\strut V[p\isbeatenby q]\right|$.
\item[(ii)]  For any $p'\not\in\conv(V)$, there is a point $p\in \conv(V)$ with $\beta(p,V)>\beta(p',V)$.
\item[(iii)] For any $d'>d$ we have $\beta_{d'}^* \leq \beta_{d}^*$.
\end{enumerate}
\end{observation}
\begin{proof}
Note that for every point~$r\not\in\conv(V)$ there is a point $r'\in\conv(V)$ that lies
strictly closer to all voters in~$V$, namely the point $r'\in\bd\conv(V)$ closest to~$r'$.
This immediately implies part~(i): if $p$ is beaten by some point~$q\not\in\conv(V)$
then $p$ is certainly beaten by a point $q'\in\conv(V)$ that lies strictly closer to
all voters in~$V$ than~$q$. It also immediately implies part~(ii), because if a point $p$ lies strictly
closer to all voters in $V$ than a point~$p'$, then~$\beta(p,V)>\beta(p',V)$.

To prove part~(iii), let $V\in\Reals^d$ be a voter set such that $\beta(V) = \beta^*_d$.
Now embed $V$ into~$\Reals^{d'}$, say in the flat $x_{d+1}=\cdots=x_{d'}=0$,
obtaining a set $V'$. Then $\beta(V')=\beta(V)$ by parts (i) and (ii).
Hence, $\beta_{d'}^* \leq \beta(V')=\beta(V) = \beta_{d}^*$.
\end{proof}

We can now prove an upper bound on $\beta_d^*$.
\begin{lemma}\label{lem:beta-star-2d-ub}
$\beta_d^* \leq \sqrt{3}/2$, for $d\geq 2$.
\end{lemma}
\begin{proof}
  By Observation~\ref{obs:dimension-monotonicity}(iii), it suffices to prove the
  lemma for~$d=2$. To this end let $V = \{v_1,v_2,v_3\}$ consist of three voters
  that form an equilateral triangle~$\Delta$ of side length~2 in~$\Reals^2$;
  see Fig.~\ref{fig:equi-triangle}(i).
  \begin{figure}
  \centering
  \includegraphics{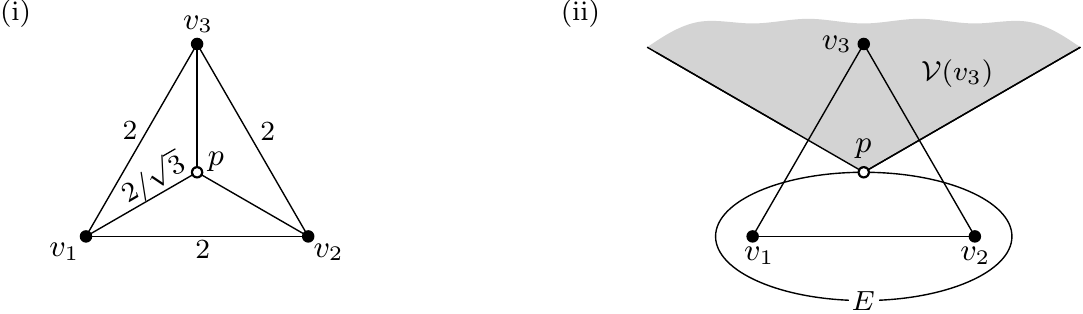}
        \caption{(i) The set $V=\{v_1,v_2,v_3\}$ of voters and the point~$p$
                     used in the proof of Lemma~\ref{lem:beta-star-2d-ub}.
                 (ii) The ellipse $E$ is tangent to the Voronoi cell $\V(v_3)$.}
        \label{fig:equi-triangle}
    \end{figure}

  Let $p$ denote the center of~$\Delta$. We will first argue that $\beta(p,V) = \sqrt{3}/2$.
  Note that $|pv_i|=2/\sqrt{3}$ for all three voters~$v_i$. Hence, for $\beta=\sqrt{3}/2$, the open balls $D_{\beta}(v_i,p)$ are pairwise disjoint and touching at the mid-points of the edges of $\Delta$.
  Therefore any competitor~$q$ either wins one voter and loses the remaining two,
  or wins no voter and loses at least one. The former happens when $q$ lies inside
  one of the three balls~$D_{\beta}(v_i,p)$; the later happens when $q$ does not lie inside any of the balls,
  because in that case $q$ can be on the boundary of at most two of the balls.
  Thus, for $\beta=\sqrt{3}/2$, the point $p$ always wins more voters than~$q$ does.
 On the other hand, for $\beta>\sqrt{3}/2$, any two balls
  $D_{\beta}(v_i,p)$, $D_{\beta}(v_j,p)$ intersect and so
  a point $q$ located in such a pairwise intersection wins two voters and beats~$p$.
  We conclude that $\beta(p,V) = \sqrt{3}/2$, as claimed.

  The lemma now follows if we can show that $\beta(p',V) \leq \sqrt{3}/2$ for any $p'\neq p$.
  Let $\VD(V)$ be the Voronoi diagram of $V$, and let $\V(v_i)$ be the closed
  Voronoi cell of~$v_i$, as shown in Fig.~\ref{fig:equi-triangle}(ii).
  Assume without loss of generality that $p'$ lies in $\V(v_3)$. Let $E$ be the ellipse
  with foci $v_1$ and~$v_2$ that passes through~$p$. Thus
  \[
  E \mydef \{ z \in \Reals^2 : |zv_1| + |zv_2| = 4/\sqrt{3} \}.
  \]
  Note that $E$ is tangent to $\V(v_3)$ at the point $p$. Hence, any point $p'\neq p$ in $\V(v_3)$ has
  $|p'v_1| + |p'v_2| > 4/\sqrt{3}$. This implies that for $\beta\geq \sqrt{3}/2$
  we have $\beta\cdot|p'v_1| + \beta\cdot|p'v_2| > 2$, and so the
  disks $D_{\beta}(p',v_1)$ and $D_{\beta}(p',v_2)$ intersect.
  It follows that for $\beta\geq \sqrt{3}/2$ there is a competitor~$q$ that
  wins two voters against~$p'$, which implies $\beta(p',V) \leq \sqrt{3}/2$
and thus finishes the proof of the lemma.
\end{proof}
We now prove lower bounds on $\beta_d^*$. We first prove that $\beta_d^*\geq 1/\sqrt{d}$
for any $d\geq 2$, and then we improve the lower bound to $\sqrt{3}/2$ for $d=2$.
The latter bound is tight by Lemma~\ref{lem:beta-star-2d-ub}.

\medskip

Let $V$ be a finite multiset of $n$ voters in $\Reals^d$. We call a hyperplane~$h$
\emph{balanced} with respect to $V$, if both open half-spaces defined by $h$ contain at most
$n/2$~voters from~$V$. Note the difference with median hyperplanes, which are required to
have fewer than $n/2$ voters in both open half-spaces.
Clearly, for any~$1\leq i\leq d$ there is a balanced hyperplane orthogonal to the $x_i$-axis,
namely the hyperplane $x_i=m_i$, where $m_i$ is a median in the multiset of all $x_i$-coordinates
of the voters in~$V$. (In fact, for any direction~$\vec{d}$ there is a balanced hyperplane
orthogonal to~$\vec{d}$.)

\begin{lemma}\label{lem:beta-star-2d-lb}
Let $d\geq 2$.
For any finite multi-set $V$ of voters in $\Reals^d$
there exists a point~$p\in\Reals^d$ such that $\beta(p,V) = 1/\sqrt{d}$.
Moreover, such a point~$p$ can be computed in $O(n)$ time.
\end{lemma}
\begin{proof}
    Let $\mathcal{H} \mydef \{h_1,\dotsc,h_d\}$ be a set of balanced hyperplanes
    with respect to $V$ such that $h_i$ is orthogonal to the $x_i$-axis, and assume
    without loss of generality that~$h_i$ is the hyperplane~$x_i=0$.
    We will prove that the point~$p$ located at the origin is a $\beta$-plurality point
    for~$V$ for any $\beta<1/\sqrt{d}$, thus showing that $\beta(p,V)\geq 1/\sqrt{d}$.

    Let $q=(q_1,\ldots,q_d)$ be any competitor of~$p$. We can assume without loss of
    generality that $\max_{1\leq i\leq d} |q_i| = q_d > 0$. Thus $q$ lies in the
    closed cone~$C_d^+$ defined as
    \[
    C_d^+ \mydef \{ \ (x_1,\ldots,x_d)\in\Reals^d : \ x_d \geq |x_j| \mbox{ for all $j\neq d$ } \}.
    \]
    Note that~$C_d^+$ is bounded by portions of the $2(d-1)$ hyperplanes $x_d = \pm x_j$ with $j\neq d$; see Fig.~\ref{fig:3D_LB}.

    Because $h_d \colon x_d=0$ is a balanced hyperplane, the open halfspace $h^+_d \colon x_d >0$ contains at most $n/2$ voters, which implies that the closed halfspace $\mycl(h_d^-) \colon x_d \leq 0$ contains at least $n/2$ voters. Hence, it suffices to argue that for any $\beta < 1/\sqrt{d}$ the point~$p$~wins all the voters in $\mycl(h_d^-)$ against~$q$.

\begin{figure}
  \centering
\includegraphics[width=0.4\textwidth]{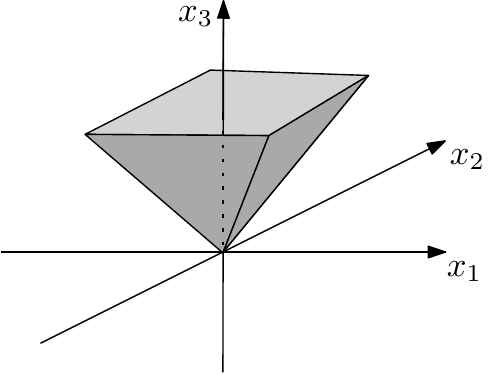}
\caption{The cone $C^+_3$ used in the proof of Lemma~\ref{lem:beta-star-2d-lb}.}
  \label{fig:3D_LB}
\end{figure}

    \begin{claiminproof}
       For any voter $v \in \mycl(h_d^{-})$ with $v\neq p$, we have that
       $\sin \left(\angle qpv\right) \geq  1/\sqrt{d}$ with equality
       if and only if $q$ lies on an edge of $C_d^+$ and $v$ lies on the orthogonal projection of this edge onto~$h_d$.
    \end{claiminproof}
    \begin{proofinproof}
       For any point~$v$ below $h_d$ there is a point~$v'\in h_d$
       with $\angle qpv' < \angle qpv$, namely the orthogonal projection
       of $v$ onto $h_d$. Hence, from now on we assume that $v\in h_d$.

        First, we prove that $\sin (\angle qpv) =  1/\sqrt{d}$ if $q$ lies on an edge~$e$ of~$C_d^+$ and $v$ lies on the orthogonal projection $\overline{e}$~of~$e$ onto~$h_d$.
        Assume without loss of generality that $e$ is the edge of~$C_d^+$ defined by the intersection of the $d-1$ hyperplanes $x_d=x_j$, so that $q_1 = \dots = q_{d-1} = q_d$. Since $\angle qpv$
        is the same for any $v\in \overline{e}$, we may assume that
        $v$ is the orthogonal projection of $q$ to $h_d$, which
        means~$|qv|=q_d$. We then have
        \[
        \sin \left(\angle{qpv}\right)
        = \frac{|qv|}{|pq|}
        = \frac{q_d}{\sqrt{q_1^2 + \dots + q_d^2}}
= \frac{1}{\sqrt{d}}\,.
        \]
        Now assume the condition for equality does not hold. Let
        $\rho$ be the ray starting at $p$ and containing~$q$, and let
        $\overline{\rho}$ be its orthogonal projection onto~$h_d$.
We have two cases:
        $v\in\overline{\rho}$ but $q$ is not contained in an edge of $C_d^+$, or $v\not\in\overline{\rho}$.

        In the former case we may, as before, assume that $v$ is
        the projection of $q$ onto~$h_d$. Since~$q\in C_d^+$
        we have $q_d \geq |q_{j}|$ for all $j$. Moreover, since
        $q$ does not lie on an edge of $C_d^+$ we have
        $q_d > |q_{j*}|$ for at least one $j^*$. Thus
        $|pq| = \sqrt{q_1^2 + \dotsc + q_d^2} < \sqrt{d \cdot q_d^2} = \sqrt{d} \cdot |qv|$,
and $\sin \left(\angle{qpv}\right)  = |qv|/|pq| > 1/\sqrt{d}$.

        In the latter case, let $\ell$ be the line containing $p$ and~$v$,
        and let $v'$ be the point on $\ell$ closest  to~$q$. Then
        $|qv| \geq |qv'| > |q\overline{q}|$, where $\overline{q}$~is the projection of $q$ onto~$h_d$, and so
\renewcommand\qedsymbol{\textcolor{darkgray}{\ensuremath{\vartriangleleft}}}
        \[
          \sin \left( \angle{qpv}\right) \geq \frac{|qv'|}{|pq|} > \frac{|q\overline{q}|}{|pq|} \geq \frac{1}{\sqrt{d}}\,.
          \qedhere
        \]
   \end{proofinproof}

We can now use the Law of Sines and the claim above to derive that
    for any $\beta<1/\sqrt{d}$ and any voter $v\in \mycl(h_d^{-})$ with $v\neq p$ we have
    \[
    \beta \cdot |pv|
    < \frac{1}{\sqrt{d}}\cdot |pv|
    = \frac{1}{\sqrt{d}} \cdot \frac{|qv| \cdot \sin\left( \angle pqv \right)}{\sin\left( \angle qpv \right)}
    \leq |qv| \cdot \sin\left( \angle pqv \right)
    \leq |qv|\,.
    \]
    Hence, $p$ wins every point in $\mycl(h_d^-)$. This proves the first part of the lemma
    since $\mycl(h_d^-)$ contains at least $n/2$ voters, as already remarked.
    \medskip

    Computing the point~$p$ is trivial once we have the balanced hyperplanes~$h_i$, which can be found
    in $O(n)$ time by computing a median $x_i$-coordinate for each $1\leq i\leq d$.
\end{proof}

\mypara{A tight bound in the plane.}
In $\Reals^2$ we can improve the above bound: for any voter set $V$ in the plane we can
find a point~$p$ with $\beta(p,V)\geq \sqrt{3}/2$. By Lemma~\ref{lem:beta-star-2d-ub},
this bound is tight.
The improvement is based on the following lemma.
\begin{figure}
  \centering
\includegraphics[width=0.4\textwidth, page=1]{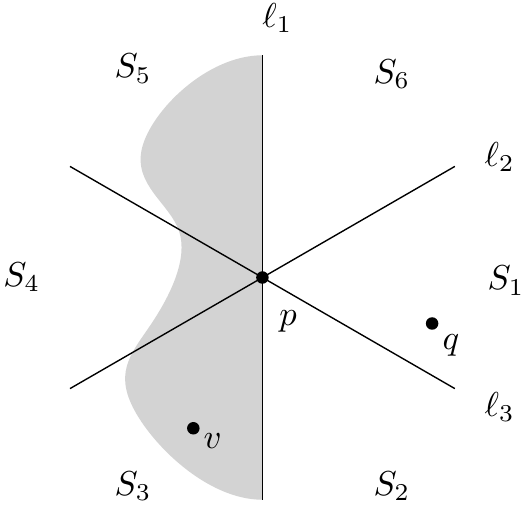}
\caption{The partition used in the proof of Lemma~\ref{lem:beta-star-2d-upper}.  The region $H = \mycl(S_3 \cup S_4 \cup S_5)$ is indicated in grey. }
  \label{fig:planar-partition}
\end{figure}
\begin{lemma}
    \label{lem:beta-star-2d-upper}
    Let $V$ be a multiset of $n$ voters in $\Reals^2$, let $\ell_1,\ell_2,\ell_3$
    be a triple of concurrent balanced lines such that the
    smaller angle between any two of them is $\frac{\pi}{3}$, and let $p$ be the
    common intersection of~$\ell_1,\ell_2,\ell_3$. Then $\beta(p,V)\geq\sqrt{3}/2$.
\end{lemma}
   \begin{proof}
    Let $q$ be a competitor of~$p$. The three lines $\ell_1,\ell_2,\ell_3$ partition
    the plane into six equal-sized sectors,  which we number $S_1$ through $S_6$ in a
    clockwise fashion, so that $q$ lies in the closure of $S_1$; see Figure~\ref{fig:planar-partition}.
Let $H$ be the closure of $S_3 \cup S_4 \cup S_5$. It is a closed halfspace bounded
    by a balanced line, so it contains at least half the voters.

    Using an analysis similar to that in the proof of Lemma~\ref{lem:beta-star-2d-lb}, we can show
    that $p$ does not lose any voter $v \in H$. Indeed, using the Law of Sines we obtain
    \[
    \frac{\sqrt{3}}{2} \cdot |pv|
    =  \frac{\sqrt{3}}{2} \cdot \frac{\sin{\angle{pqv}}}{\sin{\angle{qpv}}} \cdot |qv|
    \leq  |qv|, \quad \text{since} \quad \angle{qpv} \geq \pi/3,
\]
    which shows that $p$ is a $\beta$-plurality point for any $\beta<\sqrt{3}/2$.
    Hence, $\beta(p,V)\geq\sqrt{3}/2$.
\end{proof}
The main question is whether a triple of concurrent lines as in Lemma~\ref{lem:beta-star-2d-upper}
always exists. The next lemma shows that this is indeed the case.
The lemma---in fact a
stronger version, stating that any two opposite cones defined by the three concurrent lines
contain the same number of points---has been proved for even~$n$ by Dumitrescu~\etal~\cite{dpt-dhcnl-12}.
Our proof of Lemma~\ref{coll-bisector} is similar to their proof.
We give it because we also need it for odd~$n$, and because we will need an understanding of the proof
to describe our algorithm for computing the concurrent triple in the lemma.
Our algorithm will run in $O(n \log n)$ time, a significant improvement
over the $O(n^{4/3} \log^{1+\eps} n)$ running time obtained
(for the case of even~$n$) by Dumitrescu~\etal~\cite{dpt-dhcnl-12}.
\begin{lemma} \label{coll-bisector}
    For any multiset $V$ of $n$ voters in $\Reals^2$, there exists a triple of concurrent
    balanced lines $(\ell_1,\ell_2,\ell_3)$ such that the
    smaller angle between any two of them is $\frac{\pi}{3}$.
\end{lemma}
\begin{proof}
    Define the \emph{orientation} of a line to be the counterclockwise angle it
    makes with the positive $y$-axis. Recall that for any given orientation~$\theta$
    there exists at least one balanced line with orientation~$\theta$.
    When $n$ is odd this line is unique: it passes through the median of the
    voter set~$V$ when $V$ is projected orthogonally onto a line
    orthogonal to the lines of orientation~$\theta$.
    \begin{figure}
      \centering
      \includegraphics{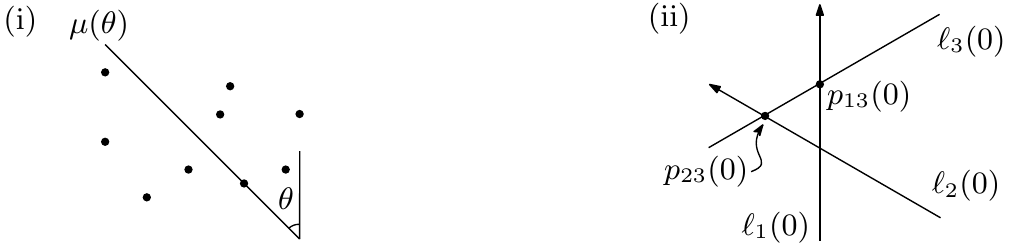}
      \caption{(i) The balanced line~$\mu(\theta)$.
      (ii) If $p_{23}$ is to the left of the directed line~$\ell_1(0)$ then $p_{13}(0)$ is to the right of $\ell_2(0)$.}
\label{fig:bisector}
    \end{figure}
    In the rest of the proof it will be convenient to have a unique balanced line
    for any orientation~$\theta$. To achieve this when $n$ is even, we simply delete an arbitrary
    voter from~$V$. (If there are other voters at the same location, these voters are not deleted.)
    This is allowed because  when $|V|$ is even, a balanced line for $V\setminus\{v\}$ is also
    a balanced line for~$V$.

    Now let $\mu$ be the function that maps an angle value~$\theta$ to the unique
    balanced line~$\mu(\theta)$; see Figure~\ref{fig:bisector}(i).
    Note that $\mu$ is continuous for $0\leq \theta<\pi$.
    Let $\ell_1(\theta) \mydef \mu(\theta)$, and $\ell_2(\theta) \mydef \mu(\theta+\frac{\pi}{3})$, and
    $\ell_3(\theta) \mydef \mu(\theta+\frac{2\pi}{3})$.
    For $i\neq j$, let $p_{ij}(\theta) \mydef \ell_i(\theta) \cap \ell_j(\theta)$
    be the intersection point between~$\ell_i(\theta)$ and~$\ell_j(\theta)$.
    If $p_{23}(0) \in \ell_1(0)$ then the lines $\ell_1(0), \ell_2(0), \ell_3(0)$
    are concurrent and we are done. Otherwise, consider the situation at $\theta=0$
    and imagine $\ell_1(0)$ and $\ell_2(0)$ to be directed in the positive $y$-direction,
    as in Fig.~\ref{fig:bisector}(ii).
    Clearly, if $p_{23}(0)$ is to the left of the directed line~$\ell_1(0)$
    then  $p_{13}(0)$ is to the right of the directed line~$\ell_2(0)$, and vice versa.
    Now increase $\theta$ from~$0$ to~$\pi/3$, and note that $\ell_1(\pi/3)=\ell_2(0)$
    and $p_{23}(\pi/3)=p_{13}(0)$. Hence,~$p_{23}(\theta)$ lies to a different side
    of the directed line $\ell_1(\theta)$ for $\theta=0$ than it does for $\theta=\pi/3$.
    Since both $\ell_1(\theta)$ and $p_{23}(\theta)$ vary continuously with~$\theta$, this
    implies that, for some $\bar{\theta}\in (0,\pi/3)$, the point~$p_{23}(\bar{\theta})$
    crosses the line~$\ell_1(\bar{\theta})$, and so
    the lines~$\ell_1(\bar{\theta}),\ell_2(\bar{\theta}),\ell_3(\bar{\theta})$
    are concurrent.
\end{proof}

The previous two lemmas show that any voter set $V$ in $\Reals^2$ admits a point~$p$
such that $\beta(p,V)\geq\sqrt{3}/2$.
We now show that we can compute such a point in $O(n\log n)$ time, namely, we show
how to compute a triple as in Lemma~\ref{coll-bisector} in $O(n\log n)$ time.
We follow the definitions and notation from the proof of that lemma.
We will assume that $n$ is odd, which, as argued, is without loss of generality.
\medskip

To find a concurrent triple of balanced lines, we first compute the lines
$\ell_1(0), \ell_2(0), \ell_3(0)$ in $O(n)$ time. If they are concurrent,
we are done. Otherwise, there is a $\bar{\theta}\in (0,\pi/3)$ such that
$\ell_1(\bar{\theta}),\ell_2(\bar{\theta}),\ell_3(\bar{\theta})$ are concurrent. To find this
value~$\bar{\theta}$ we dualize the voter set~$V$, using the standard duality transform that
maps a point~$(a,b)$ to the non-vertical line~$y=ax+b$ and vice versa. Let $v^*$ denote the dual line of
the voter~$v$, and let $V^*\mydef\{v^*:v\in V\}$. Note that for~$\theta\in (0,\pi/3)$ the
lines~$\ell_1(\theta),\ell_2(\theta),\ell_3(\theta)$ are all non-vertical, therefore their duals~$\ell^*_i(\theta)$ are well-defined.

Consider the arrangement~$\A(V^*)$ defined by the duals of the voters.
For $\theta\neq 0$, define $\slope(\theta)$ to be the slope of the
lines with orientation~$\theta$. Then $\mu^*(\theta)$, the dual of~$\mu(\theta)$,
is the intersection point of the vertical line~$x=\slope(\theta)$
with $L_{\med}$,  the median level in $\A(V^*)$.
(The median level of $\A(V^*)$ is the set of points~$q$ such that there are fewer than $n/2$~lines
below~$q$ and fewer than $n/2$~lines above~$q$; this is well defined since we assume $n$ is odd. The median level forms an $x$-monotone polygonal curve along edges of $\A(V^*)$.)
Observe that the duals~$\ell_1^*(\theta),\ell_2^*(\theta),\ell_3^*(\theta)$ all lie
on $L_{\med}$. For $\theta\in(0,\pi/3)$,
the $x$-coordinate of $\ell_1^*(\theta)$ lies in~$(-\infty,-1/\sqrt{3})$,
the $x$-coordinate of $\ell_2^*(\theta)$ lies in~$(-1/\sqrt{3},1/\sqrt{3})$,
and the $x$-coordinate of $\ell_3^*(\theta)$ lies in~$(1/\sqrt{3},\infty)$.
\begin{figure}[b]
\centering
\includegraphics{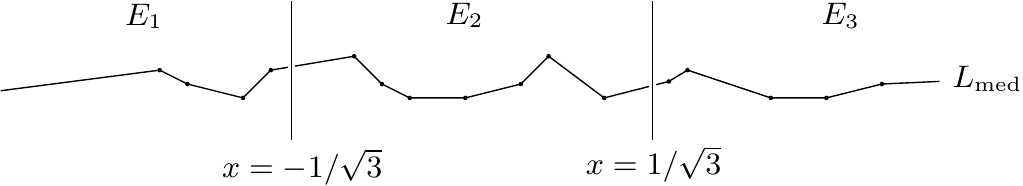}
\caption{The edge sets $E_1$, $E_2$, and $E_3$ of $L_\med$, the median level in $\A(V^*)$.}
\label{fig:dual}
\end{figure}
We split $L_\med$ into three pieces corresponding to these ranges of~the~$x$-coordinate.
Let~$E_1$, $E_2$, and $E_3$ denote the sets of edges forming the parts of $L_\med$
in the first, second, and third range, respectively, where edges crossing the vertical
lines $x=-1/\sqrt{3}$ and $x=1/\sqrt{3}$ are split; see Fig.~\ref{fig:dual}.
Thus, for any~$\theta\in(0,\pi/3)$ and for~$i\in\{1,2,3\}$,
the point~$\ell^*_i(\theta)$ lies on an edge in $E_i$.
\medskip

Recall that we want to find a value~$\bar{\theta}\in(0,\pi/3)$ such that
$\ell_1(\bar{\theta}),\ell_2(\bar{\theta}),\ell_3(\bar{\theta})$
are concurrent.
For $-\infty <x<1/\sqrt{3}$, let $\theta_x$ be such that $\slope(\theta_x)=x$,
and for $i\in\{1,2,3\}$ define $p_i(x) := \ell^*_i(\theta_{x})$.
We are thus looking for a value $\bar{x}\in(-\infty,1/\sqrt{3})$
such that the three points
$p_1(\bar{x}),p_2(\bar{x}),p_3(\bar{x})$ are collinear.

One way to find~$\bar{x}$ would be to first explicitly compute $L_\med$; then we can increase $x$,
starting at $x=-\infty$, and see how the points $p_i(x)$ move over $E_i$, until we reach a value $\bar{x}$
such that $p_1(\bar{x}),p_2(\bar{x}),p_3(\bar{x})$ are collinear.
Since the best known bounds on the complexity of the median level
is~$O(n^{4/3})$~\cite{d-ibpks-98} we will proceed differently, as follows.

\begin{enumerate}
\item \label{alg:find-interval}
    Use the recursive algorithm described below to find an interval $\bar{I}\subseteq \left(-\infty,1/\sqrt{3}\right)$ containing a value $\bar{x}$ with the desired property---namely, that the points $p_1(\bar{x})$, $p_2(\bar{x})$, and $p_3(\bar{x})$ are collinear---and such that, for any $i\in\{1,2,3\}$, the point~$p_i(x)$ lies on the same edge of~$E_i$ for all~$x\in \bar{I}$.
\item \label{alg:finish}
    Find a value $\bar{x}\in \bar{I}$ such that $p_1(\bar{x})$, $p_2(\bar{x})$, and $p_3(\bar{x})$ are collinear. Since for $i=1,2,3$ each $p_i(x)$ lies on a fixed edge~$e_i$
    of~$L_\med$ for all $x\in\bar{I}$ after Step~\ref{alg:find-interval}, this can be done in
    $O(1)$ time. Indeed, if we go back to primal space we are given three (not necessarily distinct) voters $v_1,v_2,v_3$
    (namely, the primals of the lines containing $e_1$, $e_2$, and $e_3$)
    and a range $(\theta,\theta')$ of angles (corresponding to the $x$-range~$\bar{I}$),
    such that the line $\ell_i(\theta)$ passes through $v_i$ for any $\theta\in (\theta,\theta')$.
    We then only have to compute an orientation $\bar{\theta}\in (\theta,\theta')$
    such that the lines $\ell_i(\theta)$ meet in a common point---such an orientation $\bar{\theta}$
    is guaranteed to exist by our construction of~$\bar{I}$.
\end{enumerate}

We now explain the recursive algorithm used in Step~\ref{alg:find-interval}.
In a generic call we are given three trapezoids $\Delta_1,\Delta_2,\Delta_3$
that are each bounded by two non-vertical edges and two vertical edges (one of
which may degenerate into a point).
Let $I_i$ be the $x$-range of $\Delta_i$, for $i\in\{1,2,3\}$; note that this
is well-defined since $\Delta_i$ is a trapezoid delimited by two vertical edges.
The trapezoids $\Delta_1,\Delta_2,\Delta_3$ will have the following properties.
\\
\begin{itemize}
    \item[(P1)] Trapezoid $\Delta_1$ lies inside the vertical slab~$(-\infty,-1/\sqrt{3})\times (-\infty,\infty)$,
          trapezoid~$\Delta_2$ lies inside the vertical slab~$(-1/\sqrt{3},1/\sqrt{3})\times (-\infty,\infty)$, and
          trapezoid~$\Delta_3$ lies inside the vertical slab~$(1/\sqrt{3},\infty)\times (-\infty,\infty)$.
          Moreover, the $x$-ranges $I_1,I_2,I_3$ correspond to each other in the following sense.
          Recall that an $x$-range $I\subset (-\infty,-1/\sqrt{3})$ in the dual plane corresponds to an
          angular interval $(\phi,\phi')$ in the primal plane. We denote by $I\oplus\theta$
          the $x$-range corresponding to the angular interval $(\phi+\theta,\phi'+\theta)$.
          The $x$-ranges $I_1,I_2,I_3$ will be such that $I_2=I_1\oplus\frac{\pi}{3}$ and $I_3=I_1\oplus\frac{2\pi}{3}$.
    \item[(P2)] For any $i\in\{1,2,3\}$, the part of the median level $L_\med$ inside the vertical slab $I_i \times (-\infty,\infty)$
           lies entirely inside $\Delta_i$. Together with the property~(P1) this implies that
           for any $x\in I_1$ we have that $p_i(x) \in \Delta_i$, for $i\in\{1,2,3\}$.
    \item[(P3)] Let $x_{\myleft}$ and $x_\myright$ be such that $I_1 = (x_{\myleft},x_\myright)$.
          Then we have: If $p_{1}(x_{\myleft})$ lies above the line through
          $p_2(x_{\myleft})$ and $p_3(x_{\myleft})$,
          then $p_{1}(x_{\myright})$ lies below the line through
          $p_2(x_{\myright})$ and $p_3(x_{\myright})$, and vice versa.
          This guarantees that there exists a value $\bar{x}\in I_1$
          with the desired property and, hence, that $I_1$ contains
          the interval $\bar{I}$ we are looking for. \\
\end{itemize}
In a recursive call we are also given for each $\Delta_i$ the sets $V^*_i\subseteq V^*$
of lines intersecting the interior of~$\Delta_i$, as well as $n_i^-$, the number of lines from
$V^*$ passing completely below~$\Delta_i$.

Initially, $\Delta_1$ is the unbounded
trapezoid\footnote{Since $x_{\myleft}=-\infty$ for this initial trapezoid~$\Delta_1$,
    the points $p_i(x_{\myleft})$ are not well defined. Recall, however, that in the primal plane
    the point on $L_\med$ ``at $x=-\infty$'' corresponds to the vertical balanced line $\ell_1(0)$.
    Hence, we can derive the relative position of $p_{1}(-\infty)$ with respect to the line through
    $p_2(-\infty)$ and $p_3(-\infty)$, from the relative position of the intersection point
    $\ell_2(0)\cap\ell_3(0)$ with respect to~$\ell_1(0)$.}
$(-\infty,-1/\sqrt{3})\times (-\infty,\infty)$. Similarly,
we initially have $\Delta_2 = (-1/\sqrt{3},1/\sqrt{3})\times (-\infty,\infty)$ and
$\Delta_3 = (1/\sqrt{3},\infty)\times (-\infty,\infty)$.
Furthermore,  $V^*_i=V^*$ and $n_i^-=0$ for~$i=1,2,3$.

The recursion ends when the interior of each $\Delta_i$ is intersected
by a single edge of~$L_\med$; we then have $\bar{I} := I_1$.
The recursive call for a given triple $\Delta_1,\Delta_2,\Delta_3$ starts by shrinking
$\Delta_1$ to a trapezoid $\Delta'_1$---thus zooming
in on the value~$\bar{x}$---as follows. (We assume $|V^*_1|>1$,
otherwise the shrinking of~$\Delta_1$ can be skipped.).
\begin{figure}[b]
\centering
\includegraphics{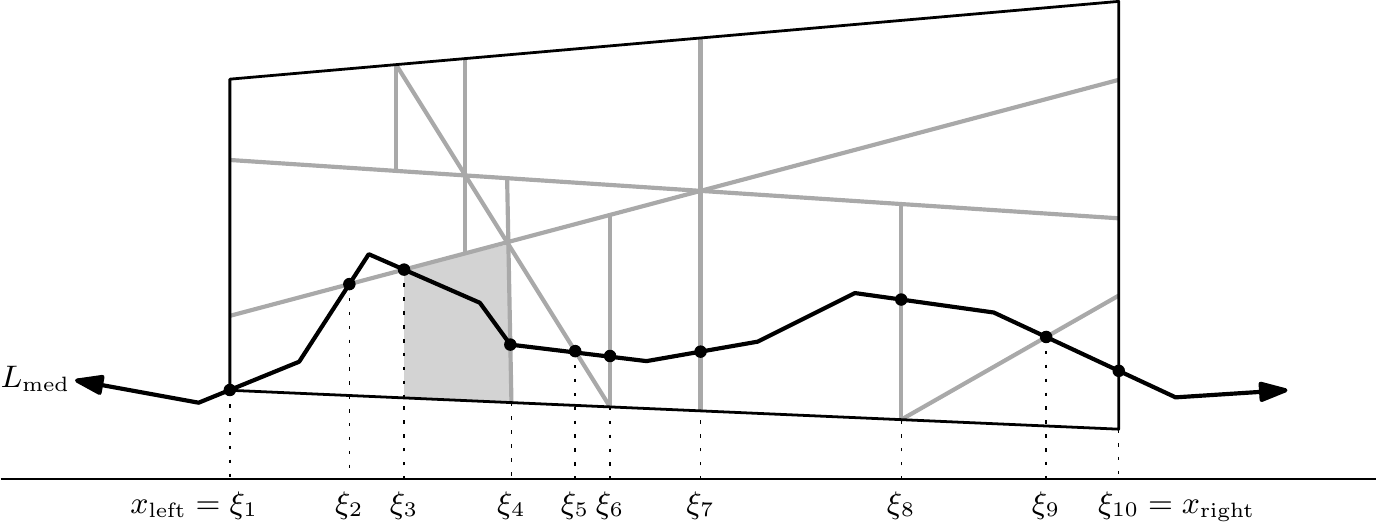}
\caption{The trapezoid $\Delta_1$ (depicted in black) and the cutting~$\Xi$ (depicted in grey).}
\label{fig:cutting}
\end{figure}
\\
\begin{enumerate}
    \item[(i)] Set $r:=2$ and construct a $(1/r)$-cutting for the lines in $V^*_1$, clipped
          to within~$\Delta_1$. In other words, construct a partition $\Xi$ of $\Delta_1$  into $O(r^2)=O(1)$
          smaller trapezoids---see Fig.~\ref{fig:cutting}---such that the interior of each trapezoid $\tau \in \Xi$ is intersected
          by at most $n_1/2$ lines from $V^*_1$, where $n_1 := |V^*_1|$.
          Computing~$\Xi$ can be done in $O(n_1)$ time~\cite{chazelle-cutting}.
    \item[(ii)] Compute the intersections of $L_\med$ with the edges of each trapezoid~$\tau\in\Xi$
          in $O(n_1\log n_1)$ time, as follows.
          \begin{itemize}
              \item Consider a non-vertical edge $e$ of $\tau$.
                    First, compute the level of $p_{\myleft}$, the left endpoint of~$e$. This can be done
                    by counting the number of lines from $V^*_1$ below $p_\myleft$ and adding $n_1^-$ to this number.
                    Next, intersect~$e$ with all lines of $V^*_1$ and sort the intersections along~$e$,
                    distinguishing lines that cross~$e$ ``from above'' and ``from below.''
                    Finally, walk along $e$, starting from $p_\myleft$ towards the right,
                    increasing and decreasing the level according to the type of the intersection we encounter. All intersection points that lie on $L_\med$ can thus be reported.
                    (For simplicity we ignore the case where $e$ partially or fully overlaps $L_\med$.
                    This can either be handled by a simple modification of the procedure,
                    or we can avoid the situation altogether by modifying the cutting such
                    that no edge $e$  of the cutting is contained in an input line.)
                \item For a vertical edge $e$ of $\tau$, we proceed similarly: first compute the level
                      of the lower endpoint of~$e$, and then walk upward along $e$ until we reach
                      an intersection point at the median level, or the upper endpoint of~$e$.
\end{itemize}
    \item[(iii)] The previous step gives us all intersection points
          of $L_\med$ with the edges of trapezoids in $\Xi$. Let $X=\{\xi_1,\ldots \xi_{|X|}\}$ be
          the sorted set of all $x$-coordinates of these intersection points, as illustrated in Fig.~\ref{fig:cutting}. Note that $|X|=O(n_1)$.
          We perform a binary search on~$X$ to find two consecutive $x$-coordinates,
          $\xi_i$ and $\xi_{i+1}$, such that the interval $(\xi_{i},\xi_{i+1})$ contains
          a value~$\bar{x}$ with the desired property. Each step in the binary search can be done in $O(n_1+n_2+n_3)$ time, as follows.

          Assume without loss of generality that $p_{1}(x_{\myleft})$ lies above the line through
          $p_2(x_{\myleft})$ and $p_3(x_{\myleft})$, and that $p_{1}(x_{\myright})$ lies below the line through
          $p_2(x_{\myright})$ and $p_3(x_{\myright})$. Suppose that during the binary search we arrive
          at some $\xi_j\in X$, and we want to decide if we want to proceed to the left or to the
          right of~$\xi_j$. To do so, we first compute the points $p_1(\xi_j)$, $p_2(\xi_j)$,
          and $p_3(\xi_j)$. Point $p_1(\xi_j)$ can be computed in $O(n_1)$ time, as follows:
          first compute all intersections of the vertical line~$x=\xi_j$ with the lines
          in $V^*_1$, and then find the intersection point whose $y$-coordinate has the appropriate rank,
          taking into account that there are $n_1^-$ lines from $V^*\setminus V^*_1$
          fully below~$\Delta_1$. The points $p_2(\xi_j)$ and $p_3(\xi_j)$ can be computed in the
          same way---this takes $O(n_2)$ and $O(n_3)$ time, respectively---after determining their $x$-coordinates in $O(1)$ time. (These $x$-coordinate
          are $\slope(\theta_{\xi_j}+\pi/3)$ and $\slope(\theta_{\xi_j}+2\pi/3)$,
          respectively.)
          After computing $p_1(\xi_j)$, $p_2(\xi_j)$, and $p_3(\xi_j)$ we can make our decision:
          we proceed to the left if $p_{1}(\xi_j)$ lies below the line through
          $p_2(\xi_j)$ and $p_3(\xi_j)$, and to the right if $p_{1}(\xi_j)$ lies above that line. (In the fortunate situation that
          $p_1(\xi_j),p_2(\xi_j),p_3(\xi_j)$ are collinear, we can take $\bar{x} \coloneqq \xi_j$
          and we are done.)

          Since each step in the binary search takes $O(n_1+n_2+n_3)$ time, the total
          binary search takes $O((n_1+n_2+n+3)\log n_1)$ time. Sorting the set $X$
          before the binary search only increases this by a constant factor.
    \item[(iv)] Finally, we take the two $x$-coordinates $\xi_i,\xi_{i+1}$ computed in the previous step,
          and find the points where $L_\med$ crosses the vertical lines $x=\xi_i$ and $x=\xi_{i+1}$.
          Between $x=\xi_i$ and $x=\xi_{i+1}$ we know that $L_\med$ lies inside a single
          trapezoid $\tau\in\Xi$. We then intersect $\tau$ with the slab $(\xi_i,\xi_{i+1})\times (-\infty,\infty)$, to obtain the trapezoid~$\Delta'_1$. (If, for example, we would have $(\xi_i,\xi_{i+1})=(\xi_3,\xi_4)$
          in Fig.~\ref{fig:cutting}, then $\Delta'_1$ is the grey
          trapezoid.)
          Note that the number of lines
          crossing $\Delta'_1$ is at most $n_1/2$ and that the $x$-range $I'_1$ of $\Delta'_1$
          satisfies property~(P3).
\end{enumerate}
\medskip

After shrinking $\Delta_1$ in this manner, we proceed as follows. We first clip $\Delta_2$
such that its $x$-range corresponds to $I_1\oplus\frac{\pi}{3}$, and then we shrink the clipped
trapezoid $\Delta_2$ to a new trapezoid~$\Delta'_2$
in the same way as we shrunk~$\Delta_1$ to $\Delta'_1$. Thus~$\Delta'_2$ is crossed by
at most $n_2/2$ lines and $I'_2 \oplus \frac{-\pi}{3}$ satisfies property~(P3), where
$I'_2$ is the $x$-range of $\Delta'_2$. Next, we clip the $x$-range of $\Delta_3$
to $I'_2\oplus\frac{\pi}{3}$, and then we apply the shrinking procedure to~$\Delta_3$
to obtain a new trapezoid~$\Delta'_3$. Finally, we clip $\Delta'_1$
and $\Delta'_2$ so that their $x$-ranges correspond to $I'_3\oplus\frac{-2\pi}{3}$
and $I'_3\oplus\frac{-\pi}{3}$, respectively.
We then recurse on the triple~$\Delta'_1,\Delta'_2,\Delta'_3$,
passing along the appropriate sets $V^*_i$ and updating the counts~$n_i^-$.

The total time spent in all three shrinking steps is $O((n_1+n_2+n_3)\log (n_1+n_2+n_3))$,
and each $n_i$ halves at every level in the recursion. Hence, the total time
for Step~\ref{alg:find-interval} on page~\pageref{alg:find-interval} is $O(n\log n)$. As already mentioned,
Step~\ref{alg:finish} takes only constant time.
We can conclude that we can find a collinear triple~$p_1(\bar{x}),p_2(\bar{x}),p_3(\bar{x})$ in $O(n\log n)$ time. In the primal this corresponds to a triple of collinear
concurrent lines as in Lemma~\ref{lem:beta-star-2d-upper}, so
we obtain following theorem.
\begin{theorem}
    \label{thm:beta-star-summary}
    \begin{enumerate}
    \item[]
    \item[(i)]
        We have $\beta_2^*=\sqrt{3}/2$. Moreover, for any multiset $V$ of $n$ voters
        in $\Reals^2$ we can compute a point~$p$ with $\beta(p,V)\geq \sqrt{3}/2$ in $O(n \log n)$ time.
    \item[(ii)]
        For $d\geq 3$, we have $1/\sqrt{d} \leq \beta_d^* \leq \sqrt{3}/2$. Moreover, for any
        multiset $V$ of $n$ voters in $\Reals^d$ with $d\geq 2$, we can compute a point~$p$ with
        $\beta(p,V)\geq 1/\sqrt{d}$ in $O(n)$ time.
    \end{enumerate}
\end{theorem}

\section{Finding a point that maximizes $\texorpdfstring{\beta}{\textbeta}(p,V)$} \label{se:maximize_beta}
We know from Theorem~\ref{thm:beta-star-summary} that, for any multiset $V$ of $n$
voters in $\Reals^d$, we can compute a point~$p$ with $\beta(p,V)\ge 1/\sqrt{d}$
(even with $\beta(p,V)\ge\sqrt{3}/2$, in the plane).
However, a given voter multiset~$V$ may admit a $\beta$-plurality point for larger values of~$\beta$---possibly even for $\beta=1$. In this section we study the
problem of computing a point $p$ that maximizes $\beta(p,V)$,
that is, a point~$p$ with $\beta(p,V)=\beta(V)$.

\subsection{An exact algorithm}
\label{sec:exact}
Below we sketch an exact algorithm to compute $\beta(V)$ together with a point~$p$ such
that $\beta(p,V)=\beta(V)$. Our goal is to show that, for constant~$d$, this can be
done in polynomial time. We do not make a special effort to optimize the exponent in
the running time; it may be possible to speed up the algorithm, but it seems
clear that it will remain impractical,  because of the asymptotic running time,
and also because of algebraic issues.

Note that we can efficiently check whether a true plurality point exists
(i.e., $\beta=1$ can be achieved) in time $O(n \log n)$ by an algorithm of De~Berg~\etal~\cite{bgm-facpp-18},
and if so, identify this point.  Therefore, hereafter $\beta=1$ is used as a sentinel value,
and our algorithm proceeds on the assumption that $\beta(p,V)<1$ for any point~$p$.
\medskip

For a voter $v\in V$, a candidate $p \in \Reals^d$, and an alternative candidate $q\in \Reals^d$,
define $f_v(p,q) \mydef \min(|qv|/|pv|,1)$ when $p\neq v$, and define $f_v(p,q) \mydef 1$ otherwise.
Observe that for $f_v(p,q)<1$ we have
\begin{itemize}
    \item $q$ wins voter~$v$ over $p$ if and only if $\beta>f_v(p,q)$,
    \item $q$ and $p$ have a tie over voter~$v$ if and only if $\beta=f_v(p,q)$, and
    \item $p$ wins voter~$v$ over $q$ if and only if $\beta<f_v(p,q)$.
\end{itemize}
For $f_v(p,q)=1$ this is not quite true: when $p=q=v$ we always have a tie, and when
$|pv|<|qv|$ then $p$ wins~$v$ even when $\beta=f_v(p,q)=1$. When $p=q$ there is a tie
for all voters, so the final conclusion (namely that
$\left |\strut V[p\succ_\beta q] \right | \geq \left|\strut V[p\isbeatenby q]\right|$)
is still correct. The fact that we incorrectly conclude that there is a tie
when $|pv|<|qv|$ and $\beta=f_v(p,q)=1$ does not present a problem either, since
we assume $\beta(p,V)<1$. Hence, we can pretend that checking
if $\beta>f_v(p,q)$, or $\beta=f_v(p,q)$, or $\beta<f_v(p,q)$ tells us whether
$q$ wins~$v$, or there's a tie, or $p$ wins~$v$, respectively.

Hereafter we identify $f_v \colon \Reals^{2d} \to \Reals$ with its
graph~$\{(p,q,f_v(p,q)) \} \subset \Reals^{2d+1}$, which is a $d$-dimensional surface. Let $f_v^+$
be the set of points lying above this graph, and $f_v^-$
be the set of points lying below it. Thus $f_v^+$ is precisely the set of combinations
of $(p,q,\beta)$ where $q$~wins~$v$ over~$p$, while $f_v$ is the set where $p$~ties
with~$q$, and $f_v^-$ is the set where $q$~loses~$v$ to~$p$.
Consider the arrangement $\A\mydef\A(F)$ defined by the set of surfaces $F \mydef \{f_v : v\in V\}$.
Each face~$C$ in~$\A$ is a maximal connected set of points with the property
that all points of~$C$ are contained in, lie below, or lie above, the same subset
of surfaces of $F$. (Note that we consider faces of all dimensions, not just
full-dimensional cells.)
Thus for all $(p,q,\beta)\in C$, exactly one of the following holds:
$\left |\strut V[p\succ_\beta q] \right | < \left|\strut V[p\isbeatenby q]\right|$, or
$\left |\strut V[p\succ_\beta q] \right | = \left|\strut V[p\isbeatenby q]\right|$, or
$\left |\strut V[p\succ_\beta q] \right | > \left|\strut V[p\isbeatenby q]\right|$.
Let $L$ be the union of all faces $C$ of $\A(F)$ such that
$\left |\strut V[p\succ_\beta q] \right | < \left|\strut V[p\isbeatenby q]\right|$, that is,
such that $p$ loses against~$q$ for all $(p,q,\beta)$ in~$C$.
We can construct $\A$ and $L$ in time $O(n^{2d+1})$ using standard machinery,
as $\A$ is an arrangement of degree-4 semi-algebraic surfaces of constant
description complexity~\cite{bcko-cgaa-2008,bpr-arag-06}.
We are interested in the set
\[
  W \mydef
  \{ (p,\beta) : \left |\strut V[p\succ_\beta q] \right | \geq \left|\strut V[p\isbeatenby q]\right|
                     \mbox{for any competitor $q$ } \}
  \subset \Reals^{d+1}.
\]
What is the relationship between $W$ and $L$?  A point $(p,\beta)$ is in $W$
precisely when, for every choice of~$q\in \Reals^d$, $p$ wins at least as many voters as~$q$
(for the given $\beta$).
In other words,
\[
  W=\{(p,\beta) \mid \text{there is no $q$ such that $(p,q,\beta)\in L$} \}.
\]
That is, $W$ is the complement of the projection of $L$ to the space~$\Reals^{d+1}$ representing the pairs~$(p,\beta)$.
The most straightforward way to implement the projection would involve constructing semi-algebraic formulas describing individual faces and invoking quantifier elimination on the resulting formulas \cite{bpr-arag-06}.
Below we outline a more obviously polynomial-time alternative.

Construct the \emph{vertical decomposition} $\vd(\A)$ of~$\A$, which is a refinement
of $\A$ into pieces (``subfaces'' $\tau$), each bounded by at most $2(2d+1)$ surfaces of
constant degree and therefore of constant complexity; see Appendix~\ref{app:vd}.
A vertical decomposition is specified
by ordering the coordinates---we put the coordinates corresponding to $q$ \emph{last}.
Since $\vd(\A)$ is a refinement of $\A$, the set~$L$ is the union of subfaces $\tau$ of $\vd(\A)$
fully contained in $L$.  Since $\A$ is an arrangement of $n$ well-behaved surfaces in $2d+1\geq 5$ dimensions,
the complexity of $\vd(\A)$ is $O(n^{2(2d+1)-4+\eps})=O(n^{4d-2+\eps})$,
for any $\eps>0$ \cite{Koltun-04}. In particular, $L$ comprises $\ell\mydef O(n^{4d-2+\eps})$ subfaces.

Since each $\tau \subset L$ is a subface of the vertical decomposition $\vd(\A)$
in which the last $d$~coordinates correspond to~$q$, the projection $\tau'$ of $\tau$
to $\Reals^{d+1}$ is easy obtain (see Appendix~\ref{app:vd}) in constant time; indeed it can be obtained by discarding the constraints on these last $d$~coordinates from the description of~$\tau$.
Thus, in time $O(\ell)$ we can construct the family of all the projections of the $\ell$~subfaces of~$L$, each a constant-complexity semi-algebraic object in $\Reals^{d+1}$.
We now construct the arrangement~$\A'$ of the resulting collection and its vertical decomposition $\vd(\A')$.   The complexity of $\vd(\A')$ is either $O(\ell^{d+1+\eps})$ or $O(\ell^{2(d+1)-4+\eps})=O(\ell^{2d-2+\eps})$, depending on whether $d+1\leq 4$ or not, respectively \cite{Koltun-04}.  Each subface in $\vd(\A')$ is either fully contained in the projection of $L$ or fully disjoint from it.  Collecting all of the latter subfaces, we obtain a representation of $W$ as a union of at most $O(\ell^{O(d)})=O(n^{O(d^2)})$ constant-complexity semi-algebraic objects.

Now if $(p,\beta) \in W$ is the point with the highest value of $\beta$,
then $\beta(V)=\beta(p,V)=\beta$.
It can be found by enumerating all the subfaces of $\vd(\A')$ contained in the closure of~$W$---we
take the closure because  $V(p,\beta)$ is defined as a supremum---and identifying their topmost point or points.
Since each face has constant complexity,
this can be done in $O(1)$ time per subface.\footnote{Once again, the projection to the $\beta$ coordinate is particularly easy to obtain if, when constructing $\vd(A')$, we set the coordinate corresponding to $\beta$ \emph{first}.}
This completes our description of an $O(n^{O(d^2)})$-time algorithm to compute
the best $\beta$ that can be achieved for a given set of voters $V$,
and the candidate $p$ (or the set of candidates) that achieve this value.

\subsection{An approximation algorithm}
Since computing $\beta(V)$ exactly appears expensive, we now turn our attention to approximation
algorithms. In particular, given a voter set $V$ in $\Reals^d$ and an $\eps \in (0,1/2]$, we wish to
compute a point $\ob{p}$ such that $\beta(\ob{p},V)\geq (1-\eps)\cdot\beta(V)$.

Our approximation algorithm works in two steps. In the first step, we compute a set $P$ of
$O(n/\eps^{2d-1}\log (1/\eps))$ \emph{candidates}.
$P$ may not contain the true optimal point~$p$, but we will
ensure that $P$ contains a point $\ob{p}$ such that $\beta(\ob{p},V)\geq (1-\eps/2)\cdot\beta(V)$.
In the second step, we approximate $\beta(p',V)$ for each $p'\in P$, to find an
approximately best candidate.

\subparagraph{Constructing the candidate set~$P$.}
To construct the candidate set~$P$, we will generate, for each voter $v_i\in V$, a set~$P_i$
of $O(1/\eps^{2d-1}\log (1/\eps))$ candidate points. Our final set $P$ of candidates will
be the union of the sets $P_1, \ldots, P_n$.  Next we describe how to construct~$P_i$.
\medskip

Partition $\Reals^d$ into a set $\C$ of $O(1/\eps^{d-1})$ simplicial cones with apex at $v_i$ and
opening angle $\eps/(2\sqrt{d})$, so that for every pair of points $u$ and $u'$ in the same cone
we have~$\angle uv_{i}u'\leq \eps/(2\sqrt{d})$. We assume for simplicity (and can easily guarantee)
that no voter in $V$ lies on the boundary of any of the cones, except for $v_i$ itself and
any voters coinciding with $v_i$.
Let $\C(v_i)$ denote the set of all cones in $\C$ whose interior contains at least one voter.
For each cone~$C\in \C(v_i)$ we generate a candidate set $G_i(C)$
as explained next, and then we set $P_i \mydef \bigcup_{C\in \C(v_i)} G_i(C) \cup \{v_i\}$.

Let $d_C$ be the distance from $v_i$ to the nearest other voter (not coinciding with~$v_i$) in~$C$.
Let $A_i(C)$ be the closed spherical shell defined by the two spheres of radii~$\eps \cdot d_C$ and~$d_C/\eps$
around~$v_i$, as shown in Fig.~\ref{fi:exp-grid}(i). The open ball of radius~$\eps\cdot d_C$ is denoted by~$A_i^{\myin}(C)$,
and the complement of the closed ball of radius~$d_C/\eps$ is denoted by~$A_i^{\myout}(C)$.
Let $G_i(C)$ be the vertices in an exponential grid defined by a collection of spheres centered at~$v_i$, and
the extreme rays of the cones in $\C$; see Fig.~\ref{fi:exp-grid}(ii).
The spheres have radii $(1+\eps/4)^i \cdot \eps \cdot d_C$, for~$0\leq i \leq \log_{(1+\eps/4)} (1/\eps^2)=O((1/\eps) \log (1/\eps))$.
Observe that~$G_i(C)$ contains not only points in~$C$, but in the entire spherical shell~$A_i(C)$.
The set $G_i(C)$ consists of $O(1/\eps^d \log (1/\eps))$ points,
and it has the following property:
\begin{quote}
  Let $p$ be any point in the spherical shell $A_i(C)$,
 and let $p'$ be a corner of the grid cell containing $p$ and nearest to $p$.
 Then $|p'p| \leq \eps \cdot |p v_i|$. \hfill $(\ast)$
\end{quote}
To prove the property, let $q$ be the point on $pv_i$ such that $|qv_i|=|p'v_i|$.
From the construction of the exponential grid we have $|pq|\leq \frac {\eps}{4} \cdot |pv_i|$.
Since $p'$ and $q$ lie in the same cone $\angle p'v_iq\leq \frac {\eps}{2\sqrt{d}}$
and, consequently, $|p'q| \leq \frac{\eps}{2} \cdot |qv_i| \leq (1+\frac{\eps}{4}) \cdot \frac{\eps}{2} \cdot |pv_i|$.
The property is now immediate since $|pp'|\leq |pq|+|qp'| < \eps \cdot |pv_i|$.

\begin{figure}
    \centering
    \includegraphics[width=12cm]{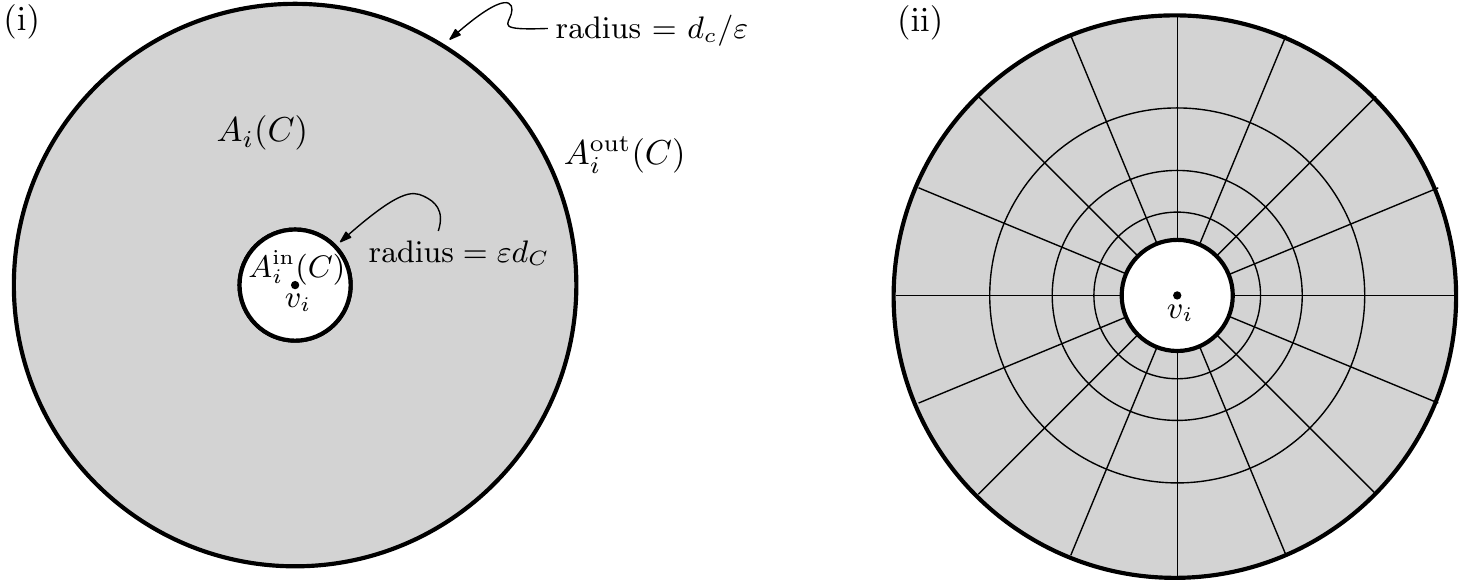}
    \caption{(i) The closed spherical shell $A_i(C)$ defined by the two balls of radii~$\eps \cdot d_C$ and~$d_C/\eps$ around~$v_i$. (ii) The exponential grid $G_i(C)$. The grid is defined by a collection of spheres centered at~$v_i$, plus extreme rays of the cones with apex at~$v_i$. The spheres have radii $(1+\eps/4)^i \cdot \eps \cdot d_C$ for $0\leq i \leq \log_{(1+\eps/4)} (1/\eps^2)=O((1/\eps) \log (1/\eps))$, and the interior angle of a cone is $\eps/2\sqrt{d}$. }
    \label{fi:exp-grid}
\end{figure}

As mentioned above, $P_i \mydef \bigcup_{C\in \C(v_i)} G_i(C) \cup \{v_i\}$,
and the final candidate set $P$ is defined as $P \mydef \bigcup_{v_i\in V} P_i$.
Computing the sets $P_i$ is easy: for each of the $O(1/\eps^{d-1})$ cones $C\in \C(v_i)$,
determine the nearest neighbor of~$v_i$ in~$C$ in $O(n)$ time by brute force,
and then generate~$G_i(C)$ in $O((1/\eps^{(d-1)})\log (1/\eps))$ time. (It is not hard to
speed up the nearest-neighbor computation using appropriate data structures,
but this will not improve the final running time in Theorem~\ref{thm:approximate_plurality_point}.)
We obtain the following lemma.
\begin{lemma}\label{lem:candidate-set-size}
  The candidate set $P$ has size $O(n/\eps^{2d-1}\log (1/\eps))$ and can be constructed in
  $O(n^2/\eps^{d-1}+n/\eps^{2d-1}\log (1/\eps))$ time.
\end{lemma}

The next lemma is crucial to show that $P$ is a good candidate set.
\begin{lemma}\label{lem:exp-grids}
For any point $p\in \Reals^d$, there exists a point $p'\in P$ with the following property:
for any voter~$v_j\in V$, we have that $|p' v_j| \leq (1+2\eps)\cdot |pv_j|$.
\end{lemma}
\begin{proof}
Let $v_i$ be a voter nearest to~$p$. We will argue that the set $P_i$ contains a point~$p'$ with the desired property.
We distinguish three cases.
\\[5mm]
\emph{Case~I: There is a cone $C\in \C(v_i)$ such that $p$ lies in the spherical shell~$A_i(C)$.}
  In this case we pick $p'$ to be a point of $G_i(C)$ nearest to~$p$, that is, $p'$ is a corner
  nearest to~$p$ of the grid cell containing $p$.
  By property~$(\ast)$ we have
  \[
  |p'v_j| \leq |p'p| + |p v_j|
          \leq \eps\cdot |pv_i|  + |pv_j|
          \leq (1+\eps)\cdot |pv_j|,
  \]
  where the last inequality follows from the fact that $v_i$ is a voter nearest to~$p$.
\\[2mm]
\emph{Case~II: Point~$p$ lies in $A_i^{\myin}(C)$ for all $C\in \C(v_i)$.}
  In this case we pick $p'\mydef v_i$. Clearly $|p' v_j| = 0 \leq (1+\eps)\cdot |pv_j|$
  for~$j=i$.
  For $j\neq i$, we argue as follows. Let $C\in\C(v_i)$ be the cone containing~$v_j$.
  Since we are in Case~II we know that $p\in A_i^{\myin}(C)$, and so
  \begin{equation} \label{eq1}
  |p'v_j| \leq |p'p| + |p v_j|
            \leq \eps d_C  + |pv_j|
            \leq \eps |p' v_j| + |pv_j|.
  \end{equation}
  Moreover, we have
  \begin{equation} \label{eq2}
  |pv_j| \geq |p' v_j| - |pp'|
         \geq |p' v_j| - \eps d_C
         \geq |p' v_j|/2,
  \end{equation}
  where the last step uses that $\eps\leq 1/2$ and $d_C \leq |p' v_j|$.
  Combining \eqref{eq1} and \eqref{eq2},
we obtain $|p' v_j| \leq (1+2\eps)\cdot |pv_j|$.
\\[2mm]
\emph{Case~III: Cases~I~and~II do not apply.}
\begin{figure}[t]
        \centering
        \includegraphics[width=10cm]{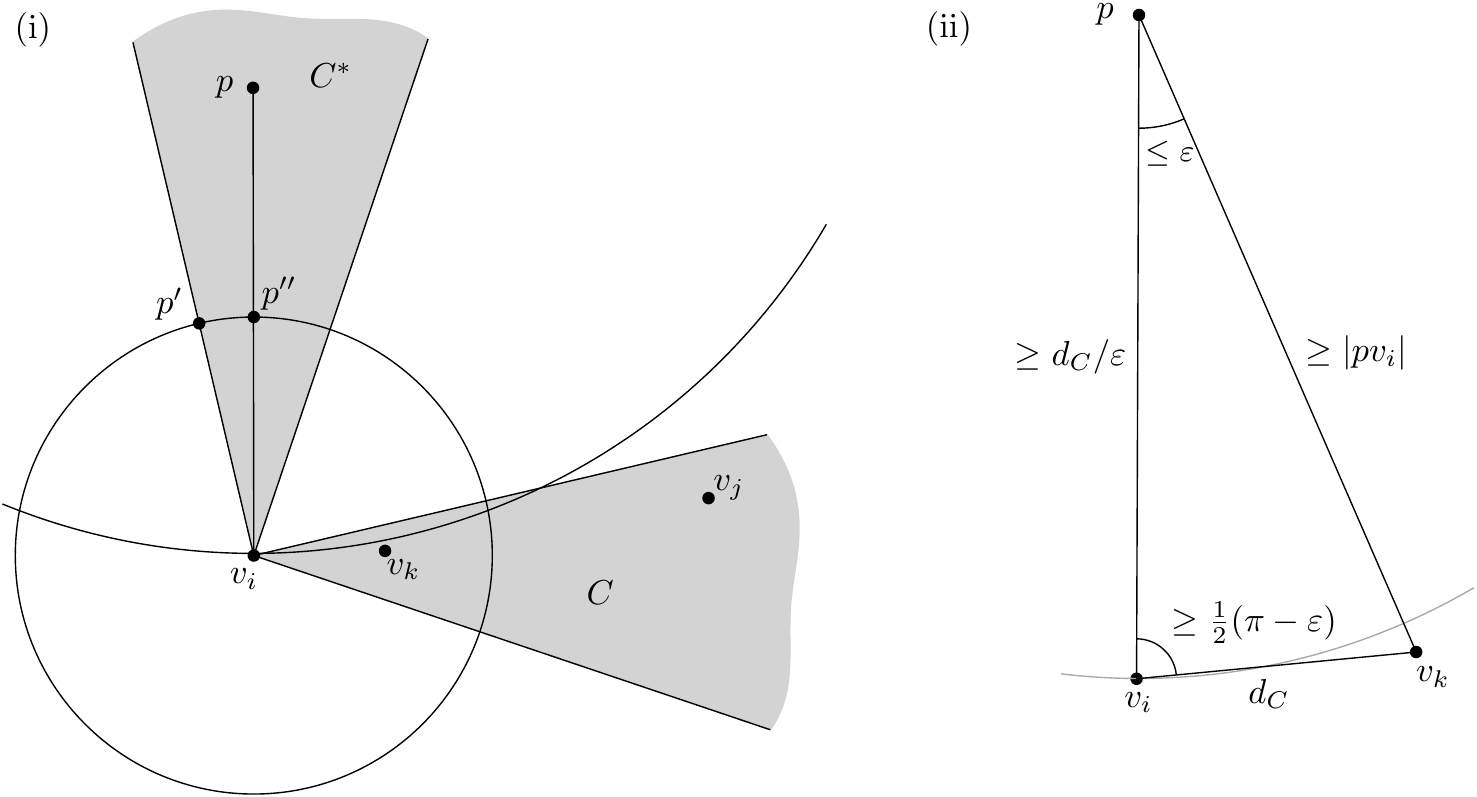}
        \caption{Illustration for Case~III.}
        \label{fig:ApxAlg}
    \end{figure}
In this case there is at least one cone $C$ such that $p\in A_i^{\myout}(C)$.
  Of all such cones, let $C^*$ be
  the one whose associated distance $d_{C^*}$ is maximized. Let $p''$ be the point on the segment $pv_i$ at distance $d_C/\eps$ from $v_i$. Without loss of generality, we will assume that $p$ and $v_i$ only differ in the $x_d$ coordinate;  see Fig.~\ref{fig:ApxAlg}(i).

  We will prove that the point $p'$ of $G_i(C^*)$ nearest to $p''$ (refer to Fig.~\ref{fig:ApxAlg}(i)) has the desired property. Consider a voter~$v_j$. We distinguish three cases.
  \begin{itemize}
  \item When $i=j$, then we have
        \[
        |p' v_i| \leq |p' p''| + |p'' v_i| \leq (1+\eps) |p'' v_i| \leq (1+\eps)|p v_i|,
        \]
        where the second inequality follows from ($\ast$).
    \item When $v_j$ lies in a cone~$C$ such that $p\in A_i^{\myin}(C)$,
        then we can use the same argument as in Case~II to show that    $|p' v_j| \leq (1+2\eps)\cdot |pv_j|$.
      \item In the remaining case $v_j$ lies in a cone~$C$ such that $p\in A_i^{\myout}(C)$.
Let $v_k$ be a voter in $C$ nearest to~$v_i$.
        Since $|v_iv_k|=d_C$, $|pv_i| \geq d_C/\eps$, and $|pv_k|\geq |pv_i|$, we can deduce that $\angle{pv_iv_k} \geq \pi/2-\eps/2$, as illustrated in Fig.~\ref{fig:ApxAlg}(ii).
        Furthermore, since $v_k$ and $v_j$ belong to the same cone $C$ the angle $\angle{v_kv_iv_j}$ is bounded by $\varepsilon/2\sqrt{d}\leq \varepsilon/2$ according to the construction. Putting the two angle bounds together we conclude that $\angle{pv_iv_j}\geq \frac{\pi}{2}-\eps$.
        Now consider the triangle defined by $p, v_i$ and $v_j$. From the Law of Sines we obtain
        \[
          \frac{|v_iv_j|}{\sin \angle{v_ipv_j}}=\frac{|pv_j|}{\sin \angle{pv_iv_j}}, \text{\quad or\quad}
|v_iv_j|=|pv_j|\cdot \frac{\sin \angle{v_ipv_j}}{\sin \angle{pv_iv_j}} \leq \frac{|pv_j|}{\cos \eps} \leq (1+\eps) \cdot |pv_j|,
        \]
        for $\eps<1/2$.
        Since $p''$ lies on the line between $p$ and $v_i$ we have:
        \[
          |p''v_j| \leq \max \{|pv_j|,|v_iv_j|\} \leq (1+\eps) \cdot |pv_j|.
        \]

        Finally we get the claimed bound by noting that $|p'p''|\leq \eps \cdot |p'v_i|$ (from~(*)),
        \[
          |p'v_j| \leq |p'p''| + |p''v_j|\leq \eps \cdot |p'v_i| + (1+\eps) \cdot |pv_j| \leq (1+2\eps) \cdot |pv_j|.  \qedhere
        \]
  \end{itemize}
\end{proof}

\subparagraph{An approximate decision algorithm.}
Given a point $p$, a positive real value $\varepsilon$ and the voter multiset $V$, we say
that an algorithm~$\Alg$ is an \emph{$\eps$-approximate} decision algorithm if
\myitemize{
 \item $\Alg$ answers \textsc{yes} if $p$ is a $\beta$-plurality point, and
 \item $\Alg$ answers \textsc{no} if $p$ is not a $(1-\eps)\beta$-plurality point.
}
In the remaining cases, where $(1-\eps)\beta < \beta(p,V) < \beta$, $\Alg$ may answer \textsc{yes} or \textsc{no}.
\medskip

Next we propose an $\eps$-approximate decision algorithm~$\Alg$.
The algorithm will use the so-called Balanced Box-Decomposition (BBD) tree
introduced by Arya and Mount~\cite{am-ars-00}. BBD trees are hierarchical
space-decomposition trees such that each node~$\mu$ represents a region in~$\Reals^d$,
denoted by~$\region(\mu)$, which is a $d$-dimensional
axis-aligned box or the difference of two such boxes. A BBD tree for a set $P$ of
$n$~points in $\Reals^d$ can be built in $O(n \log n)$ time using $O(n)$ space.
It supports $(1+\eps)$-approximate range counting queries with convex query ranges
in $O(\log n+\eps^{1-d})$ time~\cite{am-ars-00}. In our algorithm all query ranges
will be balls, hence a $(1+\eps)$-approximate range-counting query for a
$d$-dimensional ball $s(v,r)$ with center at $v$ and radius $r$ returns an integer~$I$
such that $|P \cap s(v,r)| \leq I \leq |P \cap s(v,(1+\eps)r)|$.

Our $\eps$-approximate decision algorithm $\Alg$ works as follows.
\begin{enumerate}
 \item Construct a set $Q$  of $O(n/\eps^{d-1})$ potential candidates competing
       against~$p$, as follows.  Let $Q(v)$ be a set of $O(1/\varepsilon^{d-1})$
       points distributed uniformly on the boundary of the ball $s(v,(1-\eps/2)\cdot \beta \cdot|pv|)$,
       such that the distance between any point on the boundary and its nearest
       neighbor in $Q(v)$ is at most $\frac{\eps}{4\sqrt{d}} \cdot |pv| \leq \frac{\eps}{4} \cdot \beta \cdot |pv|$. In the last step we use the fact that $\beta\geq 1/\sqrt{d}$, according to Lemma~\ref{lem:beta-star-2d-lb}.
       Set $Q \mydef Q(v_1) \cup \cdots \cup Q(v_n)$.
 \item Build a BBD tree $\tree$ on $Q$. Add a counter $c(\mu)$ to each node
       $\mu$ in $\tree$, initialized to zero.
 \item For each voter $v \in V$ perform a $(1+\eps/4)$-approximate range-counting
       query with $s(v,(1-\eps/4)\cdot \beta\cdot |pv|)$ in $\tree$. We modify the
       search in $\tree$ slightly as follows. If an internal node $\mu \in T$ is
       visited and expanded during the search, then for every non-expanded child
       $\mu'$ of $\mu$ with $\region(\mu')$ entirely contained in
       $s(v,(1+\eps/4)(1-\eps/4)\cdot \beta\cdot |pv|)) \subset s(v,\beta\cdot |pv|)$
       we increment the counter $c(\mu')$. Similarly, if a leaf is visited then
       the counter is incremented if the point stored in the leaf lies within
       $s(v,(1-\eps/4)\cdot \beta\cdot |pv|)$.
 \item For a leaf $\mu$ in $\tree$, let $M(\mu)$ be the set of nodes in $\tree$ on the
        path from the root to~$\mu$, and let $C(\mu)=\sum_{\mu' \in M(\mu)} c(\mu')$.
        Compute $C(\mu)$ for all leaves $\mu$ in $\tree$ by a
        pre-order traversal of $\tree$, and set $C \mydef \max_{\mu}C(\mu)$.
\item If $C \leq n/2$, then return \textsc{yes}, otherwise \textsc{no}.
\end{enumerate}

To prove correctness of the algorithm we define, for a given~$\gamma>0$,  a \emph{fuzzy ball}~$s_{\gamma}(v,r)$ to be any set such that
$s(v,r)\subseteq s_{\gamma}(v,r) \subseteq s(v,(1+\gamma)r)$.  Thus if $q\in s(v,r)$ then $q\in s_{\gamma}(v,r)$, if $q\not\in s(v,(1+\gamma)r)$  then $q\not\in s_{\gamma}(v,r)$, and otherwise $q$ may or may not be inside in $s_{\gamma}(v,r)$. We now observe that for each voter $v_i\in V$ there is a fuzzy ball $s_{\eps/4}(v_1,(1-\eps/4)\cdot \beta\cdot |pv_i|)$ such that the value $C(\mu)$ for a leaf $\mu$ storing a point $q$ is the depth of $q$ in the arrangement, denoted by $\A_{\eps/4}(V,1-\eps/4)$, of the fuzzy balls
$s_{\eps/4}(v_1,(1-\eps/4)\cdot \beta\cdot |pv_1|), \ldots , s_{\eps/4}(v_n,(1-\eps/4)\cdot \beta\cdot |pv|)$.

\begin{lemma} \label{lem:ApxDecAlg}
Algorithm $\Alg$ \emph{$\eps$-approximately} decides if $p$ is a $\beta$-plurality point in time $O(\frac{n}{\varepsilon^{d-1}} \log \frac{n}{\varepsilon^{d-1}})$.
\end{lemma}
\begin{proof}
We start by analyzing the running time of the algorithm. Constructing the set of points in~$Q$
can be done in time linear in $|Q|$, while building the BBD-tree~$\tree$ requires
$O((n/\varepsilon^{d-1}) \log (n/\varepsilon^{d-1}))$ time~\cite[Lemma~1]{am-ars-00}.
Next, the algorithm performs $n$ approximate range queries, each requiring
$O(\log \frac{ n}{\varepsilon^{d-1}}+\frac{1}{\varepsilon^{d-1}})$ time~\cite[Theorem~2]{am-ars-00}).
Note that the small modification we made to the query algorithm to update the counters does
not increase the asymptotic running time. Finally, the traversal of $\tree$ to compute~$C$
takes time linear in the size of $\tree$, which is $O(n/\varepsilon^{d-1})$.

\begin{figure}
        \centering
        \includegraphics[width=7cm]{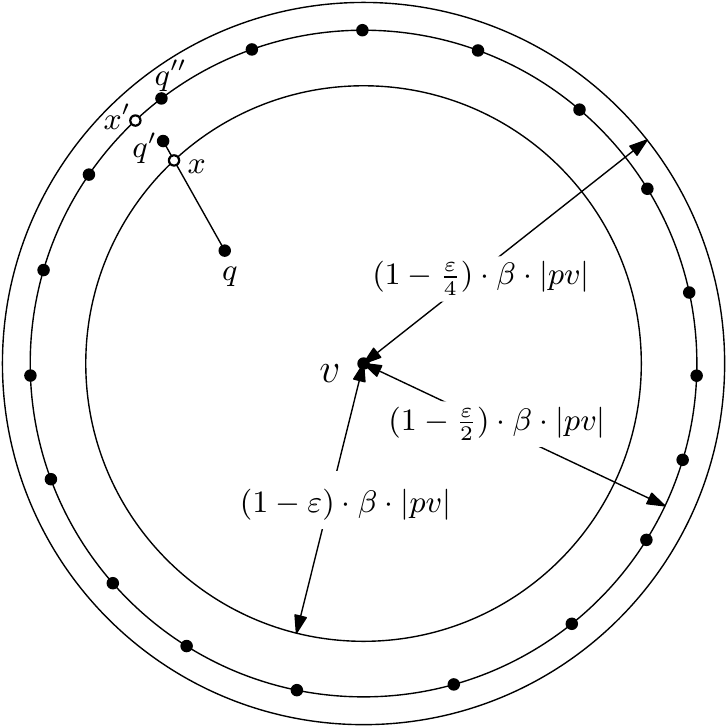}
        \caption{Illustrating the three balls of different radius used in the correctness proof of Lemma~\ref{lem:ApxDecAlg}.}
        \label{fig:dec_apx}
    \end{figure}

It remains to prove that $\Alg$ is correct.

\begin{itemize}
\item If $p$ is a plurality point there can be no point $q \in \Reals^d$
      having depth greater than $n/2$ in the arrangement of the balls $
s(v_1,\beta\cdot |pv|), \ldots , s(v_n,\beta\cdot |pv|)$.
      Since $s_{\eps/4}(v,(1-\eps/4)\cdot \beta \cdot|pv|) \subset s(v,\beta \cdot|pv|)$,
      for all $v$, $\Alg$ could not have found a point with
      depth greater than $n/2$, and hence, must return~\textsc{yes}.
\item If $p$ is not a $(1-\eps) \beta$-plurality point, then there exists a
      point $q$ with depth greater than $n/2$ in the arrangement $\A(V,1-\eps)$ of the balls $
s(v_1,(1-\eps)\cdot \beta\cdot |pv|), \ldots , s(v_n,(1-\eps)\cdot \beta\cdot |pv|)$. Let $q'$ be the point in $Q$
      nearest to $q$. We claim that for any ball $s(v,(1-\eps)\cdot \beta \cdot|pv|)$ that contains~$q$,
      its expanded version $s(v,(1-\eps/4) \cdot \beta \cdot|pv|)$ contains~$q'$.
      Of course, if $s(v,(1-\eps)\cdot \beta \cdot|pv|)$ contains $q'$ then we are done.
      Otherwise, let $x$ be the point where $qq'$ intersects the boundary of
      $s(v,(1-\eps)\cdot \beta \cdot|pv|)$; see Fig.~\ref{fig:dec_apx}.
      Note that $q'$ must also be the point in $Q$ nearest to $x$.

      Let $x'$ be the point on the boundary of $s(v,(1-\eps/2)\cdot \beta \cdot|pv|)$ nearest to $x$,
      and let $q''$ be a point in $Q$ on the boundary of $s(v,(1-\eps/2)\cdot \beta \cdot|pv|)$.
      By construction, we have
      \[
      |xx'|=\frac{\eps}{4} \cdot \beta\cdot |pv| \quad \textrm{and} \quad |x'q''|\leq \frac{\eps}{4} \cdot \beta \cdot |pv|
      \]
      and, by the triangle inequality, we obtain
      \[
      |xq'| \leq |xq''| \leq |xx'| + |x'q''| \leq  \frac{\eps}{2} \cdot \beta \cdot |pv|.
      \]
This implies that $s(v,(1-\eps/4)\cdot \beta \cdot|pv|) \subseteq s_{\eps/4}(v,(1-\eps/4)\cdot \beta \cdot|pv|)$
      must contain $q'$. Consequently, if $q$ has depth at least $n/2$ in $\A(V,1-\eps)$ then $q'$ has depth at
      least $n/2$ in the arrangement $\A_{\eps/4}(V,(1-\eps/4))$, and hence, the algorithm will return \textsc{no}.  \qedhere
\end{itemize}
\end{proof}

\subparagraph{The algorithm.}
Now we have the tools required to approximate $\beta(V)$. First, generate the set $P$ of $O(\frac{n}{\eps^{2d-1}} \log \frac {1}{\eps})$ candidate points. For each candidate point $p \in P$, perform a binary search for an approximate $\beta^*(p)$ in the interval $[1/\sqrt{d},1]$, until the remaining search interval has length at most $\eps/2 \cdot 1/\sqrt{d}$. For each $p$ and $\beta^*$, $(\eps/2)$-approximately decide if $p$ is a $\beta^*$-plurality point in $V$. Return the largest $\beta^*$ and the corresponding point $p$ on which the algorithm says \textsc{yes}.

\begin{theorem} \label{thm:approximate_plurality_point}
Given a multiset $V$ of voters in $\Reals^d$, a $((1-\eps)\cdot \beta(V))$-plurality point can be computed in $O(\frac{n^2}{\eps^{3d-2}} \cdot \log \frac{n}{\eps^{d-1}} \cdot \log^2 \frac {1}{\eps})$. \end{theorem}

\section{Concluding Remarks}
\label{se:conclusion}
We proved that any finite set of voters in $\Reals^d$ admits a $\beta$-plurality point
for $\beta=1/\sqrt{d}$ and that some sets require $\beta=\sqrt{3}/2$. For $d=2$ we managed
to close the gap by showing that $\beta^*_2=\sqrt{3}/2$. One of the main open problems is
to close the gap for $d>2$. Recall that recently the bounds for $d\geq 4$ have been
improved---see Footnote~\ref{footnote1} in the introduction---but there is still a small gap left
between the upper and lower bound.

We also presented an approximation algorithm that finds, for a given~$V$,
a $(1-\eps)\cdot\beta(V)$-plurality point. The algorithm runs in $O^*(n^2/\eps^{3d-2})$
time. Another open problem is whether a subquadratic approximation algorithm exists,
and to prove lower bounds on the time to compute $\beta(V)$ or $\beta(p,V)$ exactly.
Finally, it will be interesting to study $\beta$-plurality points in other metrics,
for instance in the personalized $L_1$-metric~\cite{bgm-facpp-18} for $d>2$ or in the $L_1$-metric for $d\geq 2$.

\section*{Acknowledgements}
The authors would like to thank Sampson Wong for improving an earlier
version of  Lemma~\ref{lem:beta-star-2d-upper}.

\bibliographystyle{plain}
\bibliography{bibliography}

\appendix

\section*{Appendix}

\section{A primer on vertical decompositions}
\label{app:vd}
We follow the notation and terminology of \cite{DS-book,vd-paper}.
A \emph{vertical decomposition} is, roughly, any partition of space into finitely many
so-called cylindrical cells (see below for a definition); it need not be a topological complex.
A \emph{vertical decomposition of an arrangement} is a refinement of an arrangement into
cylindrical cells,\footnote{The specific decomposition depends on the algorithm used
to construct it and on the ordering of the coordinates. In the computational- and combinatorial-geometry literature, one often
speaks of ``the vertical decomposition of the arrangement'' in the sense of
``the vertical decomposition obtained by applying the algorithm, say, of
Chazelle \etal \cite{vd-paper} or of Koltun \cite{Koltun-04}, to the given arrangement.''}
where \emph{refinement} means that each cylindrical cell is a subset of a face
in the arrangement.
We define cylindrical cells recursively.  To simplify the notation, any inequality limit in our definitions
can be omitted, i.e., replaced by a $\pm \infty$, as appropriate.  For example, when we talk about an open
interval $(a,b)$, i.e., the set of numbers $x$ with $a < x < b$, we include the possibilities of the
unbounded intervals $(-\infty,b)$, $(a,+\infty)$, and $(-\infty,+\infty)$.

A \emph{one-dimensional} cylindrical cell is either a singleton or an open interval $(a,b)$.
So a one-dimensional vertical decomposition is a decomposition of $\Reals$ into a finite number
of singletons and intervals.

We now define a cylindrical cell $\tau$ in $\Reals^2$.  Its projection $\tau'$ to
the $x_1$-axis is a cylindrical cell in $\Reals$.  The cell $\tau$ must have one of the following two forms:
\begin{itemize}
\item $\{(x_1,f_2(x_1)) \mid x_1 \in \tau' \}$, where
  $f_2: \tau' \to \Reals$ is a continuous total function, or
\item $\{(x_1,x_2) \mid x_1 \in \tau', f_2(x_1)<x_2<g_2(x_1) \}$, where
  $f_2,g_2: \tau' \to \Reals$ are two continuous total functions, with the property that
  $f_2(x_1)<g_2(x_1)$ for all $x_1 \in \tau'$.
\end{itemize}
If $\tau'$ is a singleton, the former defines a vertex and the latter an (open) vertical segment.  If $\tau'$ is an interval, the former defines an open monotone arc (a portion of the graph of the function $f_2$) and the latter an open \emph{pseudo-trapezoid} delimited by two (possibly degenerate)
vertical segments on left and right and by the two disjoint function graphs below and above.  (Recall that any of the limits may be omitted.  For example, $\Reals^2$ is a legal cell in a trivial two-dimensional vertical decomposition consisting only of itself, where all the limits have been ``replaced by infinities.'')

A cylindrical cell $\tau \subset \Reals^d$ is defined recursively.  Its projection $\tau'$ is a cylindrical cell in~$\Reals^{d-1}$.  Moreover, $\tau$ must have one of the following forms:
\begin{itemize}
\item $\{(x',f_d(x_1,\dots,x_{d-1})) \mid x' \in \tau' \}$, where
  $f_d: \tau' \to \Reals$ is a continuous total function, or
\item $\{(x',x_d) \mid x' \in \tau', f_d(x')<x_d<g_d(x') \}$, where
  $f_d,g_d: \tau' \to \Reals$ are two continuous total functions, with the property that
  $f_d(x')<g_d(x')$ for all $x' \in \tau'$.
\end{itemize}

A cylindrical cell is fully specified by giving its dimension and the sequence of functions~$f_i$
or pairs of functions~$f_i,g_i$, as appropriate.  In particular, the projection of the cell in
a $k$-dimensional decomposition to its first $k'<k$ coordinates can be obtained by retaining the
inequalities in the first $k'$ coordinates and discarding the remaining ones.

\end{document}